%% file: SyntacticMonoids_CALCO2015_final.tex
\documentclass[a4paper,UKenglish,numberwithinsect]{lipics}
\usepackage[textsize=footnotesize]{todonotes}
\usepackage[draft]{fixme}
\usepackage{paralist}
\usepackage{microtype}

\usepackage{ifthen}
\newboolean{full}
\setboolean{full}{true}

\include{preamble}
\title{Syntactic Monoids in a Category}
\titlerunning{Syntactic Monoids in a Category}

\author[1]{Ji\v{r}\'{i} Ad\'{a}mek}
\author[2]{Stefan~Milius\footnote{Stefan Milius acknowledges support by the Deutsche Forschungsgemeinschaft (DFG) under project MI~717/5-1}$^,$}
\author[1]{Henning Urbat}
\affil[1]{Institut f\"{u}r Theoretische Informatik \\
  Technische Universit\"{a}t Braunschweig, Germany}
\affil[2]{Lehrstuhl f\"{u}r Theoretische Informatik\\
  Friedrich-Alexander Universit\"{a}t Erlangen-N\"{u}rnberg, Germany}
\authorrunning{J.~Ad\'{a}mek, S.~Milius, H.~Urbat}
\Copyright{Ji\v{r}\'{i} Ad\'{a}mek and Stefan Milius and Henning Urbat}



\begin{document}
%
%
\FXRegisterAuthor{sm}{asm}{SM}
\FXRegisterAuthor{hu}{ahu}{HU}
\FXRegisterAuthor{ja}{aja}{JA}

\maketitle
\begin{abstract}
The syntactic monoid of a language is generalized to the level of a symmetric monoidal closed category $\D$. This allows for a uniform treatment of several notions of syntactic algebras known in the literature, including the syntactic monoids of Rabin and Scott ($\D=$ sets), the syntactic semirings of Pol\'ak ($\D=$ semilattices), and the syntactic associative algebras of Reutenauer ($\D$ = vector spaces).  Assuming that $\D$ is a commutative variety of algebras, we prove that the syntactic $\D$-monoid of a language $L$ can be constructed as a quotient of a free $\DCat$-monoid modulo the syntactic congruence of $L$, and that it is isomorphic to the transition $\D$-monoid of the minimal automaton for $L$ in $\D$. Furthermore, in the case where the variety $\DCat$ is locally finite, we characterize the regular languages as precisely the languages with finite syntactic $\D$-monoids.
\end{abstract}

\section{Introduction}

One of the successes of the theory of coalgebras is that ideas from automata theory can be developed at a level of abstraction where they apply uniformly to many different types of systems. In fact, classical deterministic automata are a standard example of coalgebras for an endofunctor. And that automata theory can be studied with coalgebraic methods rests on the observation that formal languages form the final coalgebra.

The present paper contributes to a new category-theoretic view of \emph{algebraic} automata theory. In this theory one starts with an elegant machine-independent notion of language recognition: a language $L\seq X^*$ is recognized by a monoid morphism $e: X^* \to M$ if it is the preimage under $e$ of some subset of $M$. Regular languages are then characterized as precisely the languages recognized by finite monoids. A key concept, introduced by Rabin and Scott~\cite{rs59} (and earlier in unpublished work of Myhill), is the \emph{syntactic monoid} of a language $L$. It serves as a canonical algebraic recognizer of $L$, namely the smallest $X$-generated monoid recognizing $L$. Two standard ways to construct the syntactic monoid are: 
\begin{enumerate}[(1)]
\item as a quotient of the free monoid $X^*$ modulo the \emph{syntactic congruence} of $L$, which is a two-sided version of the well-known Myhill-Nerode equivalence, and 
\item as the \emph{transition monoid} of the minimal automaton for $L$. 
\end{enumerate}
In addition to syntactic monoids there are several related notions of syntactic algebras for (weighted) languages in the literature, most prominently the syntactic idempotent semirings of Pol\'ak~\cite{polak01} and the syntactic associative algebras of Reutenauer~\cite{reu80}, both of which admit constructions similar to (1) and (2). A crucial observation is that monoids, idempotent semirings and associative algebras are precisely the monoid objects in the categories of sets, semilattices and vector spaces, respectively. Moreover, these three categories are symmetric monoidal closed w.r.t. their usual tensor product. 

The main goal of our paper is thus to develop a theory of algebraic recognition in a general symmetric monoidal closed category $\D=(\D,\t,I)$. Following Goguen \cite{goguen75}, a \emph{language} in $\D$ is a morphism $L: \xs\ra Y$ where $X$ is a fixed object of inputs, $Y$ is a fixed object of outputs, and $\xs$ denotes the free $\D$-monoid on $X$. And a \emph{$\D$-automaton} is given by the picture below: it consists of an object of states $Q$, a morphism $i$ representing the initial state, an output morphism $f$, and a transition morphism $\delta$ which may be presented in its curried form $\lambda\delta$.
\begin{equation}
  \label{eq:daut}
  \vcenter{
    \xymatrix@=8pt{
      & X \t Q \ar[d]^\delta &\\
      I \ar[r]^i & Q\ar[r]^f \ar[d]^{\lambda\delta} & Y\\
      & [X,Q] &
    }
  }
\end{equation}
This means that an automaton is at the same time an \emph{algebra} $I+X\t Q\xra{[i,\delta]} Q$ for the functor $FQ= I+ X\t Q$, and a \emph{coalgebra} $Q\xra{\langle f,\lambda \delta\rangle } Y\times [X,Q]$ for the functor $TQ = Y\times[X,Q]$. It turns out that much of the  classical (co-)algebraic theory of automata in the category of sets extends to this level of generality. Thus Goguen \cite{goguen75} demonstrated that the initial algebra for $F$ coincides with the free $\D$-monoid $\xs$, and that every language is accepted by a unique minimal $\D$-automaton. We will add to this picture the observation that the final coalgebra for $T$ is carried by the object of languages $[\xs,Y]$, see Proposition \ref{prop:fincoalg}.

In Section \ref{sec:syn} we introduce the central concept of our paper, the \emph{syntactic $\D$-monoid} of a language $L:\xs\ra Y$, which by definition is the smallest $X$-generated $\D$-monoid recognizing $L$. Assuming that $\D$ is a commutative variety of algebras, we will show that the above constructions (1) and (2) for the classical syntactic monoid adapt to our general setting: the syntactic $\D$-monoid is (1) the quotient of $\xs$ modulo the syntactic congruence of $L$ (Theorem~\ref{thm:syn}), and (2) the transition $\DCat$-monoid of the minimal $\D$-automaton for $L$ (Theorem~\ref{thm:tran}). As special instances we recover the syntactic monoids of Rabin and Scott ($\D=$ sets), the syntactic semirings of Pol\'ak ($\D=$ semilattices) and the syntactic associative algebras of Reutenauer ($\D=$ vector spaces). Furthermore, our categorical setting yields new types of syntactic algebras ``for free''. For example,  we will identify monoids with zero as the algebraic structures representing partial automata (the case $\D=$ pointed sets), which leads to the \emph{syntactic monoid with zero} for a given language. Similarly, by taking as $\D$ the variety of algebras with an involutive unary operation we obtain \emph{syntactic involution monoids}.

Most of the results of our paper apply to arbitrary languages. In Section \ref{sec:rat} we will investigate \emph{$\D$-regular languages}, that is, languages accepted by $\D$-automata with a finitely presentable object of states. Under suitable assumptions on $\D$, we will prove that a language is $\D$-regular iff its syntactic $\D$-monoid is carried by a finitely presentable object (Theorem~\ref{thm:reg}). We will also derive a dual characterization of the syntactic $\DCat$-monoid which is new even in the ``classical'' case $\D=$ sets: if $\D$ is a locally finite variety, and if moreover some other locally finite variety $\C$ is dual to $\D$ on the level of finite objects, the syntactic $\D$-monoid of $L$ dualizes to the local variety of languages in $\C$ generated by the reversed language of $L$.

Due to space limitations most proofs are omitted or sketched. See \cite{ammu15c} for an extended version of this paper.

\medskip
\textbf{Related work.} 
Our paper gives a uniform treatment of various notions of syntactic algebras known in the literature \cite{polak01,rs59,reu80}. Another categorical approach to (classical) syntactic monoids appears in the work of Ballester-Bolinches, Cosme-Llopez and Rutten \cite{bcr14}. These authors consider automata in the category of sets specified by \emph{equations} or dually by \emph{coequations}, which leads to a construction of the automaton underlying the syntactic monoid of a language. The fact that it forms the transition monoid of a minimal automaton is also interpreted in that setting. In the present paper we take a more general and conceptual approach by studying algebraic recognition in a symmetric monoidal closed category $\DCat$. One important source of inspiration for our categorical setting was the work of Goguen \cite{goguen75}.

In the recent papers \cite{ammu14,ammu15} we presented a categorical view of \emph{varieties of languages}, another central topic of algebraic automata theory. Building on the duality-based approach of Gehrke, Grigorieff and Pin \cite{ggp08}, we generalized Eilenberg's variety theorem and its local version to the level of an abstract (pre-)duality between algebraic categories. The idea to replace monoids by monoid objects in a commutative variety $\DCat$ originates in this work.

When revising this paper we were made aware of the ongoing work of Bojanczyk \cite{boj15}. He considers, in lieu of commutative varieties, categories of Eilenberg-Moore algebras for an arbitrary monad on sorted sets, and defines syntactic congruences in this more general setting. Our Theorem \ref{thm:syn} is a special case of \cite[Theorem 3.1]{boj15}.

\section{Preliminaries}
\label{sec:pre}

Throughout this paper we work with deterministic automata  in a commutative variety $\D$ of algebras. Recall that a \emph{variety of algebras} is an equational class of algebras over a finitary signature. It is called \emph{commutative} (or \emph{entropic}) if, for any two objects $A$ and $B$ of $\D$, the set $\D(A,B)$ of all homomorphisms from $A$ to $B$ carries a subobject $[A,B]\monoto B^{\under{A}}$ of the product of $\under{A}$ copies of $B$. Commutative varieties are precisely the categories of Eilenberg-Moore algebras for a commutative finitary monad on the category of sets, see \cite{kock70, lin66}. We fix an object $X$ (of inputs) and an object $Y$ (of outputs) in $\DCat$.

\begin{example}\label{ex:entvar}
\begin{enumerate}
\item  $\Set$ is a commutative variety with $[A,B] = B^A$. 

\item A \emph{pointed set} $(A,\bot)$ is a set $A$ together with a chosen point $\bot\in A$. The category $\PSet$ of pointed sets and point-preserving functions is a commutative variety. The point of $[(A,\bot_A),(B,\bot_B)]$ is the constant function with value $\bot_B$.

\item An \emph{involution algebra} is a set with an involutive unary operation $x\mapsto \tl x$, i.e.~$\tl{\tl x} = x$.  We call $\tl x$ the \emph{complement} of $x$. Morphisms are functions $f$ with $f(\tl x) = \widetilde{f(x)}$. The variety $\Inv$ of involution algebras is commutative. Indeed, the set $[A,B]$ of all homomorphisms is an involution algebra with pointwise complementation: $\tl f$ sends $x$ to $\widetilde{f(x)}$.

\item All other examples we treat in our paper are varieties of modules over a semiring. Given a semiring $\S$ (with $0$ and $1$) we denote by $\SMod{\S}$ the category of all $\S$-modules and module homomorphisms (i.e.~$\S$-linear maps). Three interesting special cases of $\SMod{\S}$ are:
\begin{enumerate}
\item $\S=\{0,1\}$, the boolean semiring with $1+1=1$: the category $\JSL$ of join-semilattices with $0$, and homomorphisms preserving joins and $0$;

\item $\S=\Int$: the category $\Ab$ of abelian groups and group homomorphisms;

\item  $\S=\K$ (a field): the category $\Vect{\K}$ of vector spaces over $\K$ and linear maps.
\end{enumerate}
\end{enumerate}
\end{example}

\begin{notation}
We denote by $\Psi: \Set\ra\D$ the left adjoint to the forgetful functor $\under{\mathord{-}}:\D\ra\Set$. Thus $\Psi X_0$ is the free object of $\D$ on the set $X_0$. \takeout{We may assume that $X$ is a subset of $\under{\Psi X}$ and the universal map is the inclusion $X\monoto \under{\Psi X}$. (If $\D$ contains at least one algebra with more than one elements, then such a choice of $\Psi$ is always possible because all universal maps are injective.)}
\end{notation}

\begin{example}
\begin{enumerate}
\item We have $\Psi X_0 = X_0$ for $\D=\Set$ and $\Psi X_0 = X_0 + \{\bot\}$ for $\D=\PSet$.
\item For $\D=\Inv$ the free involution algebra on $X_0$ is $\Psi X_0 = X_0 + \tl{X_0}$ where $\tl {X_0}$ is a copy of $X_0$ (whose elements are denoted $\tl x$ for $x\in X_0$). The involution swaps the copies of $X_0$, and the universal arrow $X_0\ra X_0+\tl{X_0}$ is the left coproduct injection.

\item For $\D=\SMod{\S}$ the free module $\Psi X_0$ is the submodule of $\S^{X_0}$ on all functions $X_0\ra \S$ with finite support. Equivalently, $\Psi X_0$ consists of formal linear combinations $\sum_{i=1}^n s_ix_i$ with $s_i\in \S$ and $x_i\in X_0$. In particular, $\Psi X_0 = \Pow_f X_0$ (finite subsets of $X_0$) for $\D = \JSL$, and $\Psi X_0$ is the vector space with basis $X_0$ for $\D=\Vect{\K}$.
\end{enumerate}
\end{example}

\begin{definition}
Given objects $A$, $B$ and $C$ of $\DCat$, a \emph{bimorphism} from $A$, $B$ to $C$ is a function $f: \under{A}\times\under{B}\ra\under{C}$ such that the maps $f(a,\mathord{-}): \under{B}\ra \under{C}$ and $f(\mathord{-},b): \under{A}\ra \under{C}$ carry morphisms of $\D$ for every $a\in\under{A}$ and $b\in \under{B}$. A \emph{tensor product} of $A$ and $B$ is a universal bimorphism $t: \under{A}\times \under{B} \ra \under{A\t B}$, which means that for every bimorphism $f: \under{A}\times \under{B} \ra \under{C}$ there is a unique morphism $f': A\t B \ra C$ in $\D$ with $f'\o t = f$.
\end{definition}

\begin{theorem}[Banaschweski and Nelson \cite{bn76}]
Every commutative variety $\D$ has tensor products, making $\D=(\D,\t,I)$ with $I=\Psi 1$ a symmetric monoidal closed category. That is, we have the following bijective correspondence of morphisms, natural in $A,B,C\in \D$:
\[\begin{array}{rl}%
   f: & A\t B \ra C\\
    \hline
    \lambda f:& A\ra [B,C]
\end{array}
\]
\end{theorem}

\begin{rem}
  \label{rem:bullet}
Recall that a \emph{monoid} $(M,m,i)$ in a monoidal category $(\D,\otimes,I)$ (with tensor product $\otimes: \DCat\times\DCat\ra\DCat$ and tensor unit $I\in \DCat$) is an object $M$ equipped with a multiplication $m: M\t M\ra M$ and unit $i: I\ra M$ satisfying the usual associative and unit laws. Due to $\t$ and $I = \Psi 1$ representing bimorphisms, this categorical definition is equivalent to the following algebraic one in our setting: a \emph{$\D$-monoid} is a triple $(M,\bullet, i)$ where $M$ is an object of $\D$ and $(\under{M},\bullet, i)$ is a monoid in $\Set$ with $\bullet: \under{M}\times\under{M}\ra\under{M}$ a bimorphism of $\D$. A \emph{morphism} $h: (M,\bullet, i)\ra (M',\bullet', i')$ of $\D$-monoids is a morphism $h: M\ra M'$ of $\D$ such that $\under{h}: \under{M}\ra\under{M'}$ is a monoid morphism in $\Set$. We denote by $\DMon$ the category of $\D$-monoids and their homomorphisms. In the following we will freely work with $\D$-monoids in both categorical and algebraic disguise.
\end{rem}

\begin{example}\label{ex:dmon}
\begin{enumerate}
\item In $\Set$ the tensor product is the cartesian product, $I=\{\ast\}$, and $\Set$-monoids are ordinary monoids.
\item In $\PSet$ we have $I=\{\bot,\ast\}$, and the tensor product of pointed sets $(A,\bot_A)$ and $(B,\bot_A)$ is $A\t B = (A\setminus\{\bot_A\})\times (B\setminus\{\bot_B\}) + \{\bot\}$. $\PSet$-monoids are precisely monoids with zero. Indeed, given a $\PSet$-monoid structure on $(A,\bot)$ we have $x\bullet \bot= \bot = \bot\bullet x$ for all $x$ because $\bullet$ is a bimorphism, i.e.~$\bot$ is a zero element. Morphisms of $\Mon{\PSet}$ are zero-preserving monoid morphisms.
\item  An $\Inv$-monoid (also called an \emph{involution monoid}) is a monoid equipped with an  involution $x\mapsto \tl x$ such that $x\bullet \tl{y}= \tl{x}\bullet y = \tl{x\bullet y}$. For example, for any set $A$ the power set $\Pow A$ naturally carries the structure of an involution monoid: the involution takes complements, $\tl S = A\setminus S$, and the monoid multiplication is the symmetric difference $S\oplus T = (S\setminus T)\cup (T\setminus S)$. 
\item $\JSL$-monoids are precisely idempotent semirings (with $0$ and $1$). Indeed, a $\JSL$-monoid on a semilattice (i.e.~a commutative idempotent monoid) $(D,+,0)$ is given by a unit $1$ and a monoid multiplication that, being a bimorphism, distributes over $+$ and $0$.
\item More generally, a $\SMod{\S}$-monoid is precisely an associative algebra over $\S$: it consists of an $\S$-module together with a unit $1$ and a monoid multiplication that distributes over $+$ and $0$ and moreover preserves scalar multiplication in both components.
\end{enumerate}
\end{example}

\begin{notation}
We denote by $X^\cnum{n}$ ($n<\omega$) the $n$-fold tensor power of $X$, recursively defined by
$X^\cnum{0} = I$ and $X^\ecnum{n+1} = X \t X^\cnum{n}$.
\end{notation}

\begin{proposition}[see Mac Lane \cite{maclane}]\label{prop:freemon}
The forgetful functor $\DMon\ra\D$ has a left adjoint assigning to every object $X$ the free $\D$-monoid $\xs = \coprod_{n<\omega} X^\cnum{n}$. The monoid structure $(\xs,m_X,i_X)$ is given by the  coproduct injection $i_X: I=X^{\t 0} \to \xs$ and $m_X: \xs\t \xs\ra \xs$, where $\xs\t \xs = \coprod_{n,k<\omega} X^{\cnum{n}} \t X^{\cnum{k}}$ and $m_X$ has as its $(n,k)$-component the $(n+k)$-th coproduct injection. The universal arrow $\eta_X: X\ra \xs$ is the first coproduct injection.
\end{proposition}

\begin{proposition}\label{prop:freemon2}
The free $\D$-monoid on $X=\Psi X_0$ is $\xs = \Psi X_0^*$. Its monoid multiplication extends the concatenation of words in $X_0^*$, and its unit is the empty word $\epsilon$.
\end{proposition}

\begin{example}
\begin{enumerate}
\item In $\Set$ we have $\xs = X^*$. In $\PSet$ with $X=\Psi X_0 = X_0 + \{\bot\}$ we get $\xs = X_0^* +\{\bot\}$. The product $x\bullet y$ is concatenation for $x,y\in X_0^*$, and otherwise $\bot$.

\item In $\Inv$ with $X = \Psi X_0 = X_0 + \tl{X_0}$ we have $\xs = X_0^* + \tl{X_0^*}$. The multiplication restricted to $X_0^*$ is concatenation, and is otherwise determined by $\tl u\bullet v = \tl{uv} = u \bullet \tl v$ for $u,v\in X_0^*$.

\item In $\JSL$ with $X=\Psi X_0 = \Pow_f X_0$ we have $\xs = \Pow_f X_0^*$, the semiring of all finite languages over $X_0$. Its addition is union and its multiplication is the concatentation of languages.

\item More generally, in $\SMod{\S}$ with $X=\Psi X_0$ we get $\xs = \Psi X_0^* = \S[X_0]$, the module of all finite $\S$-weighted languages over the alphabet $X_0$. Hence the elements of $\S[X_0]$ are functions $c: X_0^*\ra \S$ with finite support, which may be expressed as polynomials $\sum_{i=1}^n c(w_i)w_i$ with  $w_i\in X_0^*$ and $c(w_i)\in \S$. The $\S$-algebraic structure of $\S[X_0]$ is given by the usual addition, scalar multiplication and product of polynomials.
\end{enumerate}
\end{example}

\begin{definition}[Goguen \cite{goguen75}]
\label{def:aut}
A \emph{$\D$-automaton} $(Q,\delta,i,f)$ consists of an object $Q$ (of states) and morphisms $\delta: X\t Q \ra Q$, $i: I\ra Q$ and $f: Q\ra Y$; see Diagram~(\ref{eq:daut}). An \emph{automata homomorphism} $h: (Q,\delta,i,f) \ra (Q',\delta',i',f')$ is a morphism $h: Q\ra Q'$ preserving transitions as well as initial states and outputs, i.e.~making the following diagrams commute:
\[  
\xymatrix@=8pt{
X\t Q \ar[d]_{X\t h}\ar[rr]^>>>>>\delta && Q \ar[d]^h \\
X\t Q' \ar[rr]_>>>>>{\delta'}  && Q' 
}
\quad
\xymatrix@=8pt{
I \ar[drr]_{i'} \ar[rr]^i && Q \ar[d]^>>>h \ar[rr]^f  && Y \\
 && Q'\ar[urr]_{f'}  &&
}
\]
\end{definition}

The above definition makes sense in any monoidal category $\D$. In our setting, since $I=\Psi 1$, the morphism $i$ chooses an \emph{initial state} in $\under{Q}$. Moreover, if  $X=\Psi X_0$ for some set $X_0$ (of inputs), the morphism $\delta$ amounts to a choice of endomorphisms $\delta_a: Q\ra Q$ for $a\in X_0$, representing transitions. This follows from the bijections
\[\begin{array}{ll}%
    \Psi X_0 \t Q \ra Q & \text{in $\D$}\\
    \hline
    \Psi X_0 \ra [Q,Q]& \text{in $\D$}\\
    \hline
     X_0 \ra \D(Q,Q) & \text{in $\Set$}%
\end{array}
\]

\begin{example}\label{ex:automata}
\begin{enumerate}
\item The classical deterministic automata are the case $\D=\Set$ and $Y=\{0,1\}$. Here $f: Q\ra \{0,1\}$ defines the set $F=f^{-1}[1]\seq Q$ of final states. For general $Y$ we get deterministic Moore automata with outputs in $Y$.

\item The setting $\D=\PSet$ with $X = X_0 + \{\bot\}$ and $Y=\{\bot,1\}$ gives \emph{partial} deterministic automata. Indeed, the state object $(Q,\bot)$ has transitions $\delta_a: (Q,\bot)\ra (Q,\bot)$ for $a\in X_0$ preserving $\bot$, that is, $\bot$ is a sink state. Equivalently, we may consider $\delta_a$ as a partial transition map on the state set $Q\setminus\{\bot\}$. The morphism $f: (Q,\bot)\ra \{\bot,1\}$ again determines a set of final states $F=f^{-1}[1]$ (in particular, $\bot$ is non-final). And the morphism $i:\{\bot,\ast\}\ra (Q,\bot)$ determines a partial initial state: either $i(\ast)$ lies in $Q\setminus \{\bot\}$, or no initial state is defined.

\item In $\D=\Inv$ let us choose $X=X_0+\tl{X_0}$ and $Y=\{0,1\}$ with $\tl 0 = 1$. An $\Inv$-automaton is a deterministic automaton with complementary states $x\mapsto \tl x$ such that (i) for every transition $p\xra{a} q$ there is a complementary transition $\tl p \xra{a} \tl{q}$ and (ii) a state $q$ is final iff $\tl q$ is non-final.

\item For $\D=\JSL$ with $X=\Pow_f X_0$ and $Y=\{0,1\}$ (the two-chain) an automaton consists of a semilattice $Q$ of states, transitions $\delta_a: Q\ra Q$ for $a\in X_0$ preserving finite joins (including $0$), an initial state $i\in Q$ and a homomorphism $f: Q\ra\{0,1\}$ which defines a \emph{prime upset} $F=f^{-1}[1]\seq Q$ of final states. The latter means that a finite join of states is final iff one of the states is. In particular, $0$ is non-final. 

\item More generally, automata in $\D=\SMod{\S}$ with $X=\Psi X_0$ and $Y=\S$ are \emph{$\S$-weighted automata}. Such an automaton consists of an $\S$-module $Q$ of states, linear transitions $\delta_a: Q\ra Q$ for $a\in X_0$, an initial state $i\in Q$ and a linear output map $f: Q\ra \S$. 

\end{enumerate}
\end{example}

\begin{rem}\label{rem:algcoalg}
\begin{enumerate}
\item An \emph{algebra} for an endofunctor $F$ of $\D$ is a pair $(Q,\alpha)$ of an object $Q$ and a morphism $\alpha: FQ\ra Q$. A \emph{homomorphism} $h: (Q,\alpha)\ra(Q',\alpha')$ of $F$-algebras is a morphism $h: Q\ra Q'$ with $h\o \alpha = \alpha'\o Fh$. Throughout this paper we work with the endofunctor $FQ = I + X\t Q$; its algebras are denoted as triples $(Q,\delta,i)$ with $\delta: X\t Q\ra Q$ and $i: I\ra Q$. Hence $\D$-automata are precisely 
$F$-algebras equipped with an output morphism $f: Q\ra Y$. Moreover, automata homomorphisms are precisely $F$-algebra homomorphisms preserving outputs.

\item Analogously, a \emph{coalgebra} for an endofunctor $T$ of $\D$ is a pair $(Q,\gamma)$  of an object $Q$ and a morphism $\gamma: Q\ra TQ$. Throughout this paper we work with the endofunctor
$ TQ = Y\times [X,Q]$; its coalgebras are denoted as triples $(Q,\tau,f)$ with $\tau: Q\ra [X,Q]$ and $f:Q\ra Y$. Hence $\D$-automata are precisely \emph{pointed} $T$-coalgebras, i.e.~$T$-coalgebras equipped with a morphism $i: I\ra Q$. Indeed, given a pointed coalgebra $I \xra{i} Q \xra{\langle f,\tau\rangle} Y\times[X,Q]$, 
the morphism $Q\xra{\tau} [X,Q]$ is the curried form of a morphism $Q\t X \xra{\cong} X\t Q \xra{\delta} Q$. Automata homomorphisms are $T$-coalgebra homomorphisms preserving initial states.
\end{enumerate}
\end{rem}

\begin{definition}
  \label{def:asso}
Given a $\D$-monoid $(M,m,i)$ and a morphism $e: X\ra M$ of $\D$, the \emph{$F$-algebra associated to $M$ and $e$} has carrier $M$ and structure 
\[ [i,\delta] = (I+X\t M \xra{I+e\t M} I+M\t M \xra{[i,m]} M).\]
\end{definition}
In particular, the $F$-algebra associated to the free monoid $\xs$ (and its universal arrow $\eta_X$) is
\[
  [i_X,\delta_X] =(I + X \t \xs \xra{I + \eta_X \t \xs} I + \xs \t \xs \xra{[i_X,m_X]} \xs).
\] 

\begin{example} 
In $\Set$ every monoid $M$ together with an ``input'' map $e: X\ra M$ determines an $F$-algebra with initial state $i$ and transitions $\delta_a = \mathord{-}\bullet e(a)$ for all $a\in X$. The $F$-algebra associated to $X^*$ is the usual automaton of words: its initial state is $\epsilon$ and the transitions are given by $w\xra{a} wa$ for $a\in X$.
\end{example}

\begin{proposition}[Goguen \cite{goguen75}]\label{prop:initalg}
For any symmetric monoidal closed category $\D$ with countable coproducts, $\xs$ is the initial algebra for $F$.
\end{proposition}
\takeout{
In more detail, for $\xs = \coprod_{n<\omega} X^{\cnum{n}}$ we have, since $\mathord{-}\t X$ preserves coproducts (being a left adjoint), that 
$X \t \xs \cong \coprod_{n\geq 1} X^{\cnum{n}}$. Then $\xs$ is an $F$-algebra due to the following isomorphism (where $I=X^{\cnum{0}}$):
\[ F\xs \cong I + X\t \xs \cong \coprod_{n<\omega} X^{\cnum{n}} = \xs. \]
For every $F$-algebra $Q$ the unique homomorphism $e_Q: \xs = \coprod_{n<\omega} X^{\cnum{n}} \ra Q$ has components $e_Q^{\cnum{n}}\ra Q$ defines by the following recursion:
\[ e_Q^{\cnum{0}}= i: I\ra Q\quad\text{and}\quad e_Q^{\ecnum{n+1}} = (X \t X^\cnum{n} \xra{X \t e_Q^\cnum{n}} X\t Q \xra{\delta} Q). \]
}
\begin{rem}\label{rem:concept}
  Given any $F$-algebra $(Q,\delta,i)$ the unique $F$-algebra homomorphism $e_Q: \xs \ra Q$ is constructed as follows: extend the morphism $\lambda \delta: X \to [Q,Q]$ to a $\D$-monoid morphism $(\lambda\delta)^+: \xs \to [Q,Q]$. Then
  \begin{equation}\label{eq:eA}
    e_Q = (\xs \cong \xs \t I \xra{(\lambda \delta)^+ \t i} [Q,Q] \t Q \xra{\ev} Q),
  \end{equation}
where $\ev$ is the `evaluation morphism', i.e.~the counit of the adjunction $- \t Q \dashv [Q,-]$. 
\end{rem}

\begin{notation}
 $\delta^\cst: \xs\t Q\ra Q$ denotes the uncurried form of $(\lambda \delta)^+: \xs\ra[Q,Q]$.
 \end{notation}

\begin{rem}
 Recall from Rutten \cite{rutten} that the final coalgebra for the functor $TQ = \{0,1\}\times Q^X$ on $\Set$ is the coalgebra $\Pow X^*\cong [X^*,\{0,1\}]$ of all  languages over $X$. Given any coalgebra $Q$, the unique coalgebra homomorphism from $Q$ to $\Pow\Sigma^*$ assigns to every state $q$ the language accepted by $q$ (as an initial state).  These observations generalize to our present setting. The object $[\xs, Y]$ of $\D$  carries the following $T$-coalgebra structure: its transition morphism $\tau_{\fc}:\fc\ra [X,\fc]$ is the two-fold curryfication of 
\[ \fc\t X \t \xs \xra{\fc\t \eta_X \t \xs} \fc\t \xs \t \xs \xra{\fc\t m_X} \fc\t \xs \xra{\ev} Y, \]
and its output morphism $f_{\fc}: \fc\ra Y$ is 
\[ f_{[\xs, Y]} = (\fc\cong \fc\t I \xra{\fc\t i_X} \fc\t \xs  \xra{\ev} Y).\] 
\end{rem}

\begin{proposition}\label{prop:fincoalg}
$\fc$ is the final coalgebra for $T$.
\end{proposition}

\begin{proof}[Proof sketch]
Given any coalgebra $(Q,\tau,f)$, let $\delta: X\t Q\ra Q$ be the uncurried version of $\tau: Q\ra [X,Q]$, see Remark \ref{rem:algcoalg}. Then the unique coalgebra homomorphism into $\fc$ is $\lambda h: Q\ra \fc$, where $h = (Q\t \xs \cong \xs \t Q \xra{\delta^\cst} Q \xra{f} Y)$.
\end{proof}

\begin{definition}[Goguen \cite{goguen75}]
 A \emph{language} in $\D$ is a morphism $L:\xs\ra Y$.
\end{definition}

Note that if $X=\Psi X_0$ (and hence $\xs=\Psi X_0^*$) for some set $X_0$, one can identify a language $L: \xs = \Psi X_0^* \ra Y$ in $\D$ with its adjoint transpose $\tl L: X_0^*\ra \under{Y}$, via the adjunction $\Psi\dashv |\mathord{-}|: \DCat\ra \Set$. In the case where $\under{Y}$ is a two-element set, $\tl L$ is the characteristic function of a ``classical'' language $L_0 \seq X_0^*$.

\begin{example}\label{ex:language}
\begin{enumerate}
\item In $\D = \Set$ (with $\xs = X^*$ and $Y = \{ 0, 1\}$) one represents $L_0\seq X^*$  by its characteristic function $L: X^*\ra\{0,1\}$.

\item In $\D=\PSet$ (with $X=X_0 + \{\bot\}$, $\xs = X_0^* + \{\bot\}$ and $Y=\{\bot,1\}$) one represents $L_0\seq X_0^*$  by its extended characteristic function $L: X_0^*+\{\bot\}\ra \{\bot,1\}$ where $L(\bot)=\bot$.

\item In $\D=\Inv$ (with $X=X_0+\tl{X_0}$, $\xs=X_0^* + \tl{X_0^*}$ and $Y=\{0,1\}$) one represents $L_0\seq X_0^*$  by $L: X_0^* + \wt{X_0^*} \to \{0,1\}$ where $L(w) = 1$ iff $w \in L_0$ and $L(\wt w) = 1$ iff $w \not\in L_0$ for all words $w \in X_0^*$.

\item In $\D=\JSL$ (with $X=\Pow_f X_0$, $\xs = \Pow_f X_0^*$ and $Y=\{0,1\}$) one represents $L_0\seq X_0^*$ by $L: \Pow_f X_0^*\ra\{0,1\}$ where $L(U)=1$ iff $U\cap L_0 \neq \emptyset$.

\item In $\D=\SMod{\S}$ (with $X=\Psi X_0$, $\xs = \S[X_0]$ and $Y=\S$) an $\S$-weighted language $L_0: X_0^*\ra \S$ is represented by its free extension to a module homomorphism
\[
L:\S[X_0^*] \to \S,\quad L\left(\sum\limits_{i =1}^{n} c(w_i)w_i \right) = \sum_{i=1}^n c(w_i)L_0(w_i).
\]
\end{enumerate}
\end{example}

\begin{definition}[Goguen \cite{goguen75}]
The language \emph{accepted} by a $\D$-automaton $(Q,\delta,i,f)$ is $L_Q = (\xs \xra{e_Q} Q \xra{f} Y )$, where $e_Q$ is the $F$-algebra homomorphism of Remark \ref{rem:concept}.
\end{definition}

\begin{example}\label{ex:langacc}
\begin{enumerate}
\item In $\D=\Set$ with $Y=\{0,1\}$, the homomorphism $e_Q: X^* \ra Q$ assigns to every word $w$ the state it computes in $Q$, i.e.~the state the automaton reaches on input $w$. Thus $L_Q(w)=1$ iff $Q$ terminates in a final state on input $w$, which is precisely the standard definition of the accepted language of an automaton. For general $Y$, the function $L_Q: X^*\ra Y$ is the behavior of the Moore automaton $Q$, i.e.~$L_Q(w)$ is the output of the last state in the computation of $w$.

\item For $\D=\PSet$ with $X=X_0+\{\bot\}$ and $Y=\{\bot,1\}$, we have $e_Q: X_0^* + \{\bot\}\ra (Q,\bot)$ sending $\bot$ to $\bot$, and sending a word in $X_0^*$ to the state it computes (if any), and to $\bot$ otherwise. Hence $L_Q: X_0^*+\{\bot\}\ra \{\bot,1\}$ defines (via the preimage of $1$) the usual language accepted by a partial automaton.

\item In $\D=\Inv$ with $X=X_0+\tl{X_0}$ and $Y= \{0,1\}$, the map 
$L_Q: X_0^*+ \tl{X_0^*} \ra \{0,1\}$ sends $w\in X_0^*$ to $1$ iff $w$ computes a final state, and it sends $\tl w\in \tl{X_0^*}$ to $1$ iff $w$ computes a non-final state. 

\item In $\D=\JSL$ with $X=\Pow_f X_0$ and $Y=\{0,1\}$, the map $L_Q: \Pow X_0^* \ra \{0,1\}$ assigns to $U\in \Pow_f X_0^*$ the value $1$ iff the computation of at least one word in $U$ ends in a final state.

\item In $\D=\SMod{\S}$ with $X=\Psi X_0$ and $Y=\S$, the map $L_Q: \S[X_0^*] \ra \S$ assigns to $\sum_{i=1}^n c(w_i)w_i$ the value $\sum_{i=1}^n c(w_i)y_i$, where $y_i$ is the output of the state $Q$ reaches on input $w_i$. Taking $Q = \S^n$ for some natural number $n$ yields a classical $n$-state weighted automaton, and in this case one can show that the restriction of $L_Q$ to $X_0^*$ is is the usual language of a weighted automaton.
\end{enumerate}
\end{example}

\begin{rem}
By Remark \ref{rem:algcoalg} every $\D$-automaton $(Q,\delta,i,f)$ is an $F$-algebra as well as a $T$-coalgebra. Our above definition of $L_Q$ was purely algebraic. The corresponding coalgebraic definition uses the unique coalgebra homomorphism
$c_Q: Q\ra \fc$ into the final $T$-coalgebra and precomposes with $i: I\ra Q$ to get a morphism $c_Q \o i: I\ra \fc$ (choosing a language, i.e.~an element of $\fc$). Unsurprisingly, the results are equal:
\end{rem}

\begin{proposition}\label{prop:lang}
The language $L_Q: \xs \ra Y$ of an automaton $(Q,\delta,i,f)$ is the uncurried form of the morphism $c_Q \o i: I\ra \fc$.
\end{proposition}

\takeout{
\begin{example}
In $\Set$ the generators simply mean that every element of $M$ is a product of elements from $e_0[X]$. In $\Inv$ this means that every element of $M$ is a product of elements from $e_0[X]$ and their complements. In $\JSL$ an $X_0$-generated $\D$-monoid is an idempontent semiring with $e_0: X_0 \ra \under{M}$ such that every element of $M$ is a sum of products of elements from $e_0[X_0]$.
\end{example}
}

\section{Algebraic Recognition and Syntactic $\boldsymbol{\D}$-Monoids}
\label{sec:syn}

In classical algebraic automata theory one considers recognition of
languages by (ordinary) monoids in lieu of automata. One key concept
is the \emph{syntactic monoid} which is characterized as the smallest
monoid recognizing a given language. There are also related concepts of
canonical algebraic recognizers in the literature, e.g. the syntactic
idempotent semiring and the syntactic associative algebra. In this section we will give a uniform account of algebraic
language recognition in our categorical setting. Our main result is the
definition and construction of a minimal algebraic
recognizer, the \emph{syntactic $\D$-monoid} of a language. 
\begin{definition}
\label{def:recog}
A $\D$-monoid morphism $e: \xs \ra M$ \emph{recognizes} the language $L: \xs \to Y$ if there exists a morphism $f: M\ra Y$ of $\D$ with $L=f\o e$.
\end{definition}

\begin{example}\label{ex:monrec}
We use the notation of Example \ref{ex:language}. 
\begin{enumerate}
\item $\D = \Set$ with $\xs = X^*$ and $Y = \{ 0, 1\}$: given a monoid $M$, a function $f: M\ra\{0,1\}$ defines a subset $F=f^{-1}[1]\seq M$. Hence a monoid morphism $e: X^*\ra M$ recognizes $L$ via $f$ (i.e.~$L=f\o e$) iff $L_0 = e^{-1}[F]$. This is the classical notion of recognition of a language $L_0\seq X^*$ by a monoid, see e.g.~Pin~\cite{pin15}.
\item $\D=\PSet$ with $X=X_0 + \{\bot\}$, $\xs = X_0^* + \{\bot\}$ and $Y=\{\bot,1\}$: given a monoid with zero $M$, a $\PSet$-morphism $f: M\ra \{\bot,1\}$ defines a subset $F=f^{-1}[1]$ of $M\setminus\{0\}$. A zero-preserving monoid morphism $e: X_0^*+\{\bot\}\ra M$ recognizes $L$ via $f$ iff $L_0 = e^{-1}[F]$.
\item $\D=\Inv$ with $X=X_0+\tl{X_0}$,  $\xs=X_0^* + \tl{X_0^*}$ and $Y=\{0,1\}$: for an involution monoid $M$ to give a morphism $f: M\ra \{0,1\}$ means to give a subset $F=f^{-1}[1]\seq M$ satisfying $m\in F$ iff $\tl{m}\not\in F$. Then $L$ is recognized by $e: X_0^*+\tl{X_0^*}\ra M$ via $f$ iff $L_0 = X_0^*\cap e^{-1}[F]$.
\item $\D=\JSL$ with $X=\Pow_f X_0$, $\xs = \Pow_f X_0^*$ and $Y=\{0,1\}$: for an idempotent semiring $M$ a morphism $f: M\ra Y$ defines  a prime upset $F=f^{-1}[1]$, see Example \ref{ex:automata}. Hence $L$ is recognized by a semiring homomorphism $e: \Pow_f X_0^* \ra M$ via $f$ iff $L_0 = X_0^*\cap e^{-1}[F]$. Here we identify $X_0^*$ with the set of all singleton languages $\{w\}$, $w\in X_0^*$. This is the concept of language recognition introduced by Pol\'ak \cite{polak01} (except that he puts $F=f^{-1}[0]$, so $0$ and $1$ must be swapped, as well as $F$ and $M\setminus F$).
\item $\D=\SMod{\S}$ with $X=\Psi X_0$, $\xs = \S[X_0]$ and $Y=\S$: given an associative algebra $M$, the language $L$ is recognized by $e: \S[X_0]\ra M$ via $f: M\ra \S$ iff $L= f\o e$. For the case where the semiring $\S$ is a ring, this notion of recognition is due to Reutenauer \cite{reu80}. 
\end{enumerate}
\end{example}

\begin{rem}\label{rem:Lalg}
\begin{enumerate}
\item Since $\DCat$ and $\DMon$ are varieties, we have the usual factorization system of regular epimorphisms ($=$ surjective homomorphisms)   and  monomorphisms ($=$ injective homomorphisms). \emph{Quotients} and \emph{subobjects} are understood w.r.t. this system.

\item By an \emph{$X$-generated $\D$-monoid} we mean a quotient $e: \xs \epito M$ in $\DMon$. For two such quotients $e_i: \xs \epito M_i$, $i = 1,2$, we say, as usual, that $e_1$ is \emph{smaller or equal to} $e_2$ (notation: $e_1 \leq e_2$) if $e_1$ factorizes through $e_2$. Note that if $X=\Psi X_0$, the free $\D$-monoid $\xs=\Psi X_0^*$ on $X$ is also the free $\D$-monoid on the set $X_0$ (w.r.t. the forgetful functor $\DMon\ra\Set$), see Proposition~\ref{prop:freemon2}.
In this case, to give a quotient $e:\xs\epito M$ is equivalent to giving a set of generators for the $\D$-monoid $M$ indexed by $X_0$ -- which is why $M$ may also be called an \emph{$X_0$-generated $\DCat$-monoid}.
\item   Let $e: \xs \epito M$ be an $X$-generated $\D$-monoid with unit $i: I \to M$ and multiplication $m: M \t M \to M$. Recall that $\eta_X: X \to \xs$ denotes the universal morphism of the free $\D$-monoid on $X$ and consider the $F$-algebra associated to $M$ and $X \xrightarrow{\eta_X} \xs \xrightarrow{e} M$ (see Definition~\ref{def:asso}). Thus, together with a given $f: M \to Y$ an $X$-generated $\D$-monoid induces an automaton $(M, \delta, i, f)$ called the \emph{derived} automaton.
\end{enumerate}
\end{rem}

\begin{lemma}
\label{lem:recog}
The language recognized by an $X$-generated $\D$-monoid $e: \xs \epito M$ via $f: M\ra Y$ is the language accepted by its derived automaton.
\end{lemma}

We are now ready to give an abstract account of syntactic algebras in our setting. In classical algebraic automata theory the syntactic monoid of a language is characterized as the smallest monoid recognizing that language. We will use this property as our definition of the syntactic $\D$-monoid.
\begin{definition}
  \label{def:syn}
  The \emph{syntactic $\D$-monoid} of language $L:\xs\ra Y$, denoted by $\Syn L$, is the smallest $X$-generated monoid recognizing $L$. 
\end{definition}
In more detail, the syntactic $\D$-monoid is an $X$-generated $\D$-monoid $e_L: \xs \epito \Syn L$ together with a morphism $f_L: \Syn{L}\ra Y$ of $\D$ such that (i) $e_L$ recognizes $L$ via $f_L$, and (ii) for every $X$-generated $\D$-monoid $e: \xs \epito M$ recognizing $L$ via $f: M\ra Y$ we have $e_L \leq e$, that is, the left-hand triangle below commutes for some $\D$-monoid morphism $h$:
\[
\xymatrix@=8pt{
  \xs
  \ar@{->>}[rr]^-e
  \ar@{->>}[rrd]_-{e_L}
  &&
  M
  \ar@{->>}[d]^>>>h
  \ar[rr]^-f 
  &&
  Y
  \\
  &&
  \Syn L
  \ar[rru]_{f_L}
}
\]
Note that the right-hand triangle also commutes since $e$ is epimorphic and $f \o e = L = f_L \o e_L$. The universal property determines $\Syn L$, $e_L$ and $f_L$ uniquely up to isomorphism. A construction of $\Syn{L}$ is given below (Construction \ref{c:syn}). We first consider a special case:

\begin{example}
In $\D = \Set$ with $Y = \{0,1\}$, the syntactic monoid of a language $L\seq X^*$ can be constructed as the quotient of $X^*$ modulo the \emph{syntactic congruence}, see e.g. \cite{pin15}:
\[
\Syn{L} = X^* / \mathord{\sim}, \qquad \text{where $u \sim v$ iff for all $x, y \in X^*$: $xuy \in L \iff xvy \in L$}.
\]
\end{example}
We aim to generalize this construction to our categorical setting. First note the following 

\begin{lemma}
  \label{lem:monadic}
  Let $\D$ be any symmetric monoidal closed category with countable co\-products. Then the forgetful functor $\DMon \to \D$ preserves reflexive coequalizers.
\end{lemma}

\begin{notation}
  Let $(M, m, i)$ be a $\D$-monoid and $x: I \to M$. We write $x \bullet -$ and $- \bullet x$ for the following morphisms, respectively:
  \[
  M \cong I \t M \xrightarrow{x \t M} M \t M \xrightarrow{m} M
  \qquad
  \text{and}
  \qquad
  M \cong M \t I \xrightarrow{M \t x} M \t M \xrightarrow{m} M.
  \]
  Recall that in our setting, where $\D$ is a commutative variety, we have $I = \Psi 1$ and so the morphism $x$ is the adjoint transpose of an element of $M$ (see Remark~\ref{rem:bullet}). In the following we shall often write $x \bullet y$, identifying $x, y: I \to M$ with their corresponding elements of $M$. 
\end{notation}
\begin{definition}
  \label{d:syn}
  The \emph{syntactic congruence} of a language $L: \xs \to Y$ is the following relation on the underlying set of $\xs$:
  \[
  E = \{(u,v) \in \xs \times \xs \mid \forall x, y \in \xs : L(x\bullet u \bullet y) = L(x\bullet v\bullet y)\}
  \]
The projection maps are denoted by $l, r: E \to \xs$. 
\end{definition}
\begin{lemma}\label{lem:canalg}
  The set $E$ carries a canonical $\D$-algebraic structure making it a $\D$-object. 
\end{lemma}

\begin{proof}[Proof sketch] Just observe that $E = \bigcap E_{x,y}$ where for fixed $x,y\in\xs$ the object $E_{x,y}$ is the kernel of the $\D$-morphism
     $\xymatrix@1{\xs \ar[r]^-{x \bullet -} & \xs \ar[r]^-{-\bullet y} & \xs \ar[r]^-L & Y}$.
\end{proof}
That the name syntactic \emph{congruence} makes sense follows from Lemma~\ref{lem:mon} below. First recall that a \emph{$\D$-monoid congruence} on a given $\D$-monoid $M$ is an equivalence relation in $\DMon$, that is, a jointly monic pair $c_1, c_2: C \to M$ of $\D$-monoid morphisms (equivalently a $\D$-submonoid $\langle c_1, c_2\rangle: C \monoto M \times M$) which is reflexive, symmetric and transitive. Congruences on $M$ are ordered as subobjects of $M\times M$, i.e.~via inclusion.
\begin{lemma}
  \label{lem:mon}
  $E$ is a $\D$-monoid congruence on $\xs$.
\end{lemma}
We can give an alternative, more conceptual, description of $E$:

\begin{lemma}
  \label{lem:ker}
  Let $l_0, r_0: K \to  \xs$ be the kernel pair of $L:\xs \to Y$ in $\D$. Then $l, r: E \to \xs$ is the \emph{largest} $\D$-monoid congruence contained in $K$. 
\end{lemma}
\begin{construction}
  \label{c:syn}
  Let $L: \xs \to Y$ be a language and $l, r: E \to \xs$ its syntactic congruence. We construct the $\D$-monoid $\Syn{L}$ as the coequalizer of $l$ and $r$ in $\DMon$:
  \[
  \xymatrix{
    E 
    \ar@<3pt>[r]^-{l}
    \ar@<-3pt>[r]_-{r}
    &
    \xs
    \ar@{->>}[r]^-{e_L}
    &
    \Syn L.
    }
  \]
\end{construction}

We need to show that $\Syn{L}$ has the universal property of Definition \ref{def:syn}, which first requires to define the morphism $f_L: \Syn{L}\ra Y$ with $L=f_L\o e_L$. To this end consider the diagram below, where $l_0$, $r_0$ is the kernel pair of $L$ and $m$ witnesses that $E$ is contained in $K$, i.e.~$l=l_0\o m$ and $r=r_0\o m$ (see Lemma \ref{lem:ker}).
\[ 
\xymatrix{
K \ar@<0.5ex>[r]^{l_0} \ar@<-0.5ex>[r]_{r_0} & \xs \ar@{->>}[dr]_{e_L} \ar[r]^L & Y\\
E \ar[u]^m \ar@<0.5ex>[ur]^{l} \ar@<-0.5ex>[ur]_{r} & & \Syn{L} \ar@{-->}[u]_{f_L}
}
\]
  By Lemma \ref{lem:monadic} the morphism $e_L$ is also a coequalizer of $l$ and $r$ in $\D$. Since $L\o l = L\o r$ by the above diagram, this yields a unique $f_L: \Syn L \to Y$ with $L = f_L \o e_L$. In other words, $\Syn L$ recognizes $L$ via $f_L$.

\takeout{ 
\begin{proposition}
  The functor $\Psi: \Set \to \D$ preserves monoids.
\end{proposition}
\begin{proof}
  This follows from the fact that $\Psi$ is strongly monoidal (i.e.~it preserves monoidal product up to canonical isormorphism). Indeed, we have that $\Psi 1$ is the unit of the monoidal product on $\D$, and $\Psi(X \times Y) \cong \Psi X \t \Psi Y$ holds because of the following natural bijections:
  \[
  \begin{array}[b]{c}
    \Psi X \t \Psi Y \to C \\
    \hline
    \Psi X \to [\Psi Y, C] \\
    \hline
    X \to \D(\Psi Y C) \\
    \hline
    X \to \Set(Y , |C|) \\
    \hline
    X\times Y \to |C| \\
    \hline
    \Psi(X \times Y) \to C
  \end{array}\qedhere
  \]
\end{proof}} 

\begin{theorem}
  \label{thm:syn}
  $\Syn L$ together with $e_L$ and $f_L$ forms the syntactic $\D$-monoid of $L$. 
\end{theorem}

\begin{proof}[Proof sketch]
This follows from the correspondence between kernel pairs and regular quotients: since $l,r: E\ra \xs$ is the largest congruence contained in the kernel pair of $L$ by Lemma \ref{lem:ker}, the coequalizer $e_L$ of $l,r$ is the smallest quotient of $\xs$ recognizing $L$.
\end{proof}
\begin{rem}
Our proof of Theorem~\ref{thm:syn} is quite conceptual and works in a general symmetric monoidal closed category $\D$ with enough structure. On this level of generality one would use Lemma~\ref{lem:ker} to \emph{define} the syntactic congruence $E$ as the largest $\D$-monoid congruence contained in the kernel of $L: \xs \to Y$. However, it is unclear whether such a congruence exists in this generality and so its existence might have to be taken as an assumption. Hence we restricted ourselves to the setting of a commutative variety $\D$.
\end{rem}

\begin{example} Using the notation of Example $\ref{ex:language}$ we obtain the following concrete syntactic algebras:
  \begin{enumerate}
  \item In $\PSet$ with $X=X_0+\{\bot\}$ and $Y=\{\bot,1\}$ the \emph{syntactic monoid with zero} of a language $L_0\seq X_0^*$ is $(X_0^*+\{\bot\}) /\mathord{\sim}$ where, for all $u,v\in X_0^*+\{\bot\}$,
  \[ u\sim v \quad\text{iff}\quad \text{for all $x,y\in X_0^*$}: xuy\in L_0 \Leftrightarrow xvy \in L_0.\]
 The zero element is the congruence class of $\bot$.
  \item In $\Inv$ with $X = X_0 + \wt{X_0}$ and $Y = \{0,1\}$ the \emph{syntactic involution monoid} of a language $L_0\seq X_0^*$ is the quotient of $X_0 + \wt{X_0^*}$ modulo the congruence $\sim$ defined for words $u,v \in X_0^*$ as follows:
  \begin{enumerate}[(i)]
  \item $u \sim v\quad\text{iff}\quad\wt u \sim \wt v\quad\text{iff}\quad\text{for all $x, y \in X_0^*$}: xuy \in L_0 \iff xvy \in L_0$;
\item
$u \sim \wt v\quad\text{iff}\quad\wt u \sim v \quad\text{iff}\quad\text{for all $x, y \in X_0^*$}: xuy \in L_0 \iff xvy \not\in L_0$.
\end{enumerate}
\item In $\SMod{\S}$  with $X = \Psi X_0$ and $Y = \S$ the \emph{syntactic associative $\S$-algebra} of a weighted language $L_0: X_0^*\ra\S$ is the quotient of $\S[X_0]$ modulo the congruence defined for $U,V \in \S[X_0]$ as follows: 
\begin{equation}
  \label{eq:syns}
  U \sim V\quad \text{iff} \quad \text{for all $x, y \in X_0^*$}: L(xUy) = L(xVy)
\end{equation}
 Indeed, since $L:\S[X_0]\ra \S$ is linear, \refeq{eq:syns} implies $L(PUQ) = L(PVQ)$ for all $P, Q \in \S[X_0]$, which is the syntactic congruence of Definition \ref{d:syn}.
\item In particular, for $\D = \JSL$ with $X=\Pow_f X_0$ and $Y=\{0,1\}$, we get  the \emph{syntactic (idempotent) semiring} of a language $L_0 \subseteq X_0^*$ introduced by Pol\'ak \cite{polak01}: it is the quotient $\Pow_f X_0^*/\mathord{\sim}$ where for $U,V\in\Pow_f X_0^*$ we have 
\[
U \sim V\quad\text{iff}\quad\text{for all $x, y \in X_0^*$}: (xUy) \cap L_0 \neq \emptyset \iff xVy \cap L_0 \neq \emptyset.
\]
\item For $\D = \Vect{\K}$ with $X=\Psi X_0$ and $Y=\K$, the \emph{syntactic $\K$-algebra} of a $\K$-weighted language $L_0: X_0^*\ra \K$ is the quotient $\K[X_0]/I$ of the $\K$-algebra of finite weighted languages modulo the ideal
\[
I = \{ V \in \K[X_0] \mid \text{for all $x,y \in X_0^*$}: L(xVy) = 0 \}.
\]
 Indeed, the congruence this ideal $I$ generates ($U\sim V$ iff $U-V \in I$) is precisely~\refeq{eq:syns}. Syntactic $\K$-algebras were studied by Reutenauer~\cite{reu80}.
\item Analogously, for $\D = \Ab$ with $X=\Psi X_0$ and $Y=\Int$, the \emph{syntactic ring} of a $\Int$-weighted language $L_0: X_0^*\ra \Int$ is the quotient $\Int[X_0]/I$, where $I$ is the ideal of all $V\in \Int[X_0]$ with $L(xVy) = 0$ for all $x, y \in X_0^*$.
\end{enumerate}
\end{example}

\section{Transition $\boldsymbol{\D}$-Monoids}
\label{sec:tran}

Here we present another construction of the syntactic $\D$-monoid of a language: it is the transition $\D$-monoid of the minimal $\D$-automaton for this language. Recall that for any object $Q$ of a closed monoidal category $\DCat$, the object $[Q,Q]$ forms a $\D$-monoid w.r.t. composition.
\begin{definition}
  \label{def:TA}
  The \emph{transition $\D$-monoid} $\T{Q}$ of an $F$-algebra $(Q,\delta, i)$ is the image of the $\D$-monoid morphism $(\lambda\delta)^+:\xs \ra [Q,Q]$ extending $\lambda\delta: X \to [Q,Q]$:
  \[
  \xymatrix@=8pt{
    \xs \ar@{->>}[dr]_{e_{\T Q}} \ar[rr]^{(\lambda\delta)^+} & & [Q,Q] \\
    & \T Q \ar@{ >->}[ur]_{m_{\T Q}} 
  }
  \]
\end{definition}

\begin{example}
\begin{enumerate}
\item In $\Set$ the transition monoid of an $F$-algebra $Q$ (i.e.~an automaton without final states)  is the monoid of all extended transition maps $\delta_w = \delta_{a_n}\o \cdots\o \delta_{a_1}: Q\ra Q$ for $w=a_1\cdots a_n\in X^*$, with unit $\id_Q = \delta_\epsilon$ and composition as multiplication.
\item In $\PSet$ with $X=X_0+\{\bot\}$ (the setting for partial automata) this is completely analogous, except that we add the constant endomap of $Q$ with value $\bot$.
\item In $\Inv$ with $X=X_0+\tl{X_0}$ we get the involution monoid of all $\delta_w$ and $\tl{\delta_w}$. Again the unit is $\delta_\epsilon$, and the multiplication is determined by composition plus the equations $x\tl y = \tl{xy} = \tl x y$.
\item In $\JSL$ with $X=\Pow_f X_0$ the transition semiring consists of all finite joins of extended transitions, i.e.~all semilattice homomorphisms of the form
$\delta_{w_1}\vee\cdots\vee \delta_{w_n}$ for  $\{w_1,\ldots,w_n\}\in\Pow_f X_0^*$.
The transition semiring was introduced by Pol\'ak \cite{polak01}.
\item In $\SMod{\S}$ with $X=\Psi X_0$ the associative transition algebra  consists of all linear maps of the form
$\sum_{i=1}^n s_i \delta_{w_i}$ with $s_i\in \S$ and  $w_i\in X_0^*$.
\end{enumerate}
\end{example}

Recall from Definition~\ref{def:aut} that a $\D$-automaton is an $F$-algebra $Q$ together with an output morphism $f: Q \to Y$. Hence we can speak of the transition $\D$-monoid of a $\D$-automaton.
\begin{proposition}\label{prop:transmon}
The language accepted by a $\D$-automaton $(Q,\delta,f,i)$ is recognized by the $\DCat$-monoid morphism $e_{\T{Q}}: \xs\epito \T{Q}$.
\end{proposition}
\begin{proof}[Proof sketch]
The desired morphism $f_{\T{Q}}: \T{Q}\ra Y$ with $L_Q=f_{\T{Q}}\o e_{\T{Q}}$ is
\[ f_{\T Q} = (\T Q \xra{m_{\T Q}} [Q,Q]\cong [Q,Q]\t I \xra{[Q,Q]\t i} [Q,Q]\t Q \xra{\ev} Q \xra{f} Y).\qedhere \]
\end{proof}

\begin{definition}\label{def:minaut}
A $\D$-automaton $(Q,\delta ,i, f)$ is called \emph{minimal} iff it is
\begin{enumerate}[(a)]
\item \emph{reachable}: the unique $F$-algebra homomorphism $\xs\ra Q$ is surjective;
\item \emph{simple}: the unique $T$-coalgebra homomorphism $Q\ra [\xs, Y]$ is injective.
\end{enumerate}
\end{definition}

\begin{theorem}[Goguen \cite{goguen75}]\label{thm:minaut}
Every language $L: \xs \ra Y$ is accepted by a minimal $\DCat$-automaton $\Min{L}$, unique up to isomorphism. Given any reachable automaton $Q$ accepting $L$, there is a unique surjective automata homomorphism from $Q$ into $\Min{L}$.
\end{theorem}

This leads to the announced construction of syntactic $\DCat$-monoids via  transition $\DCat$-monoids. The case $\DCat=\Set$ is a standard result of algebraic automata theory (see e.g. Pin \cite{pin15}), and the case $\DCat=\JSL$ is due to Pol\'ak \cite{polak01}.
\begin{theorem}\label{thm:tran}
The syntactic $\D$-monoid of a language $L: \xs\ra Y$ is isomorphic to the transition $\D$-monoid of its minimal $\D$-automaton:
\[ \Syn{L} \cong \T{\Min{L}}. \]
\end{theorem}
\begin{proof}[Proof sketch.]
Using reachability and simplicity of $\Min{L}$, one proves that the quotients $e_L: \xs\epito \Syn{L}$ and $e_{\T{\Min{L}}}: \xs\epito\T{\Min{L}}$ have the same kernel pair, namely the syntactic congruence of $L$. This implies the statement of the theorem.
\end{proof}

\section{$\boldsymbol{\DCat}$-Regular Languages}
\label{sec:rat}
Our results so far apply to arbitrary languages in $\D$. In the present section we focus on \emph{regular languages}, which in $\D=\Set$ are the languages accepted by finite automata, or equivalently the languages recognized by finite monoids. For arbitrary $\D$ the role of finite sets is taken over by finitely presentable objects. Recall that an object $D$ of $\D$ is \emph{finitely presentable} if the hom-functor $\D(D,\mathord{-}):\D\ra\Set$ preserves filtered colimits. Equivalently, $D$ is an algebra presentable with finitely many generators and relations.

\begin{definition}
A language $L: \xs\ra Y$ is called \emph{$\D$-regular} if it is accepted by some $\D$-automaton with a finitely presentable object of states.
\end{definition}
To work with this definition, we need the following

\begin{assumptions}\label{asm:finite}
We assume that the full subcategory $\D_{f}$ of finitely presentable objects of $\D$ is closed under subobjects, strong quotients and finite products.
\end{assumptions}

\begin{example}
\begin{enumerate}
\item Recall that a variety is \emph{locally finite} if all finitely presentable algebras (equivalently all finitely generated free algebras) are finite. Every locally finite variety satisfies the above assumptions. This includes our examples $\Set$, $\PSet$, $\Inv$ and $\JSL$.
\item A semiring $\S$ is called \emph{Noetherian} if all submodules of finitely generated $\S$-modules are finitely generated. In this case, as shown in \cite{bms13}, the category $\SMod{\S}$ satisfies our assumptions.
Every field is Noetherian, as is every finitely generated commutative ring, so  $\Vect{\K}$ and $\Ab=\SMod{\Int}$ are special instances. 
\end{enumerate}
\end{example}

\begin{theorem}\label{thm:reg}
For any language $L: \xs \ra Y$ the following statements are equivalent:
\begin{enumerate}[(a)]
\item $L$ is $\D$-regular.
\item The minimal $\D$-automaton $\Min{L}$ has finitely presentable carrier.
\item $L$ is recognized by some $\D$-monoid with finitely presentable carrier.
\item The syntactic $\D$-monoid $\Syn{L}$ has finitely presentable carrier.
\end{enumerate}
\end{theorem}

\begin{proof}[Proof sketch]
This follows immediately from the universal properties of $\Syn{L}$ and $\Min{L}$ and the assumed closure properties of $\DCat_f$.
\end{proof}

Just as the collection of all languages is internalized by the final coalgebra $\fc$, see Proposition \ref{prop:fincoalg}, we can internalize the regular languages by means of the \emph{rational coalgebra}.

\begin{definition}
The \emph{rational coalgebra} $\rho T$ for $T$ is the colimit (taken in the category of $T$-coalgebras and homomorphisms) of all $T$-coalgebras with finitely presentable carrier.
\end{definition}

\takeout{
\begin{rem}
\begin{enumerate}[(a)]
\item Observe that the colimit defining $\rho T$ is filtered since (i) colimits of coalgebras are formed on the level of $\D$ and (ii) $\D_f$ is closed under finite colimits.

\item In the case where the input object $X$ is finitely presentable the
  endofunctor $T$ is finitary. This implies that $\rho T$ is a
  fixpoint of $T$, see \cite{amv_atwork}. Moreover, in this case $\rho T$ is the final locally finite
  coalgebra \cite{m_linexp}, that is, every fp-coalgebra has a unique coalgebra homomorphism into $\rho T$. However, we will not use
  these properties in the present setting, so we can work with
  arbitrary input objects $X$.
\end{enumerate}
\end{rem} 
}

\begin{proposition}\label{prop:ratfix}
There is a one-to-one correspondence between $\D$-regular languages and elements $I\ra \rho T$ of the rational coalgebra.
\end{proposition}

We conclude this section with an interesting dual perspective on syntactic monoids, based on our previous work \cite{ammu14,ammu15}. For lack of space we restrict to the case $\D=\Set$. This category is \emph{predual} to the category $\BA$ of boolean algebras in the sense that the full subcategories of finite sets and finite boolean algebras are dually equivalent. Indeed, this is a restriction of the well-known Stone duality: the dual equivalence functor assigns to a finite boolean algebra $B$ the set $\At(B)$ of its atoms, and to a boolean homomorphism $h: A\ra B$  the map $\At(h): \At(B)\ra\At(A)$ sending $b\in \At(B)$ to the unique atom $a\in\At(A)$ with $ha\geq b$.

 How do the concepts we investigated in $\Set$ -- languages, automata and monoids -- dualize to $\BA$? Observe that $\Reg(X)$, the boolean algebra of regular languages over the alphabet $X$, can be viewed as a deterministic automaton: its final states are the regular languages containing the empty word, and the transitions are given by $L\xra{a} a^{-1}L$ for $a\in X$, where $a^{-1}L = \{w\in X^*: aw\in L\}$ is the \emph{left derivative} of $L$ w.r.t. the letter $a$. (Similarly, the \emph{right derivative} of $L$ w.r.t. $a$ is $La^{-1} = \{w\in X^*: wa\in L\}$.)
This makes $\Reg(X)$ a coalgebra for the endofunctor $\ol T =\{0,1\}\times \Id^X$ on $\BA$. Since the two-chain $\{0,1\}$ is dual to the singleton set $1$, finite coalgebras for $\ol T$ dualize to finite \emph{algebras} for the functor $F=1+X\times \Id \cong 1 + \coprod_X \Id$ on $\Set$. Based on this, we proved in \cite{ammu14} that further (i) finite $\ol{T}$-subcoalgebras of $\Reg(X)$ dualize to finite quotient algebras of the initial $F$-algebra $X^*$, and (ii) finite \emph{local varieties of languages} (i.e.~finite $\ol{T}$-subcoalgebras of $\Reg(X)$ closed under right derivatives) dualize to those $F$-algebras associated to $X$-generated monoids, see Definition \ref{def:asso}. For a regular language $L\seq X^*$ the $F$-algebras associated to the minimal automaton $\Min{L}$ and the syntactic monoid $\Syn{L}$ are finite. Their dual $\ol{T}$-coalgebras are characterized as follows:

\begin{theorem}\label{thm:syndual} Let $L\seq X^*$ be a regular language, and $L^\rev$ its reversed language.
\begin{enumerate}[(a)]
\item $\Min{L}$ is dual to the smallest subcoalgebra of $\Reg(X)$ containing $L^\rev$.
\item $\Syn{L}$ is dual to the smallest local variety of languages containing $L^\rev$.
\end{enumerate}
\end{theorem}

Part (a) of this theorem adds to the recently developed dual view of minimal automata, see \cite{bkp13} and also \cite{mamu14, ammu14_2}. All the above considerations generalize from $\BA/\Set$ to arbitrary pairs $\C/\D$ of predual locally finite varieties of algebras. Examples include the self-predual varieties $\C=\D=\JSL$ and $\C=\D=\Vect{\K}$ for a finite field $\K$.

\section{Conclusions and Future Work}
\label{sec:con}

We proposed the first steps of a categorical theory of algebraic language recognition. Despite our assumption that $\DCat$ is a commutative variety, the bulk of our definitions, constructions and proofs works in any symmetric monoidal closed category with enough structure. However, the construction of the syntactic monoid via the syntactic congruence, and the proof that it coincides with a transition monoid, required the concrete algebraic setting. It remains an open problem to develop a genuinely abstract framework for our theory. In particular, such a generalized setting should provide the means for incorporating \emph{ordered} algebras, e.g. the syntactic ordered monoids of Pin \cite{pin15}. We expect this can be achieved by working with (order-)enriched categories, where the coequalizer in our construction of the syntactic monoid is replaced by a coinserter. A more general theory of recognition might also open the door to treating algebraic recognizers for additional types of  behaviors, including Wilke algebras \cite{wilke91} (representing $\omega$-languages) and forest algebras \cite{bw08} (representing tree and forest languages).

One of the leading themes of algebraic automata theory is the classification of languages in terms of their syntactic algebras. For instance, by Sch\"utzenberger's theorem a language is star-free iff its syntactic monoid is aperiodic. We hope that our conceptual view of syntactic monoids (notably their dual characterization in Theorem \ref{thm:syndual}) can contribute to a duality-based approach to such results, leading to generalizations and new proof techniques.

\bibliographystyle{plain}
\bibliography{coalgebra,ourpapers}

\iffull
\clearpage
\appendix
\renewcommand\thetheorem{\thesection.\arabic{theorem}}

\section{Proofs}

This Appendix contains all proofs and additional details we omitted due to space limitations.

\subsection*{Details for Example \ref{ex:entvar}.4.}
A \emph{semiring} $\S=(S,+,\o,0,1)$ consists of a commutative monoid $(S,+,0)$ and a monoid $(S,\o,1)$ such that
\[ 0x = x0 = 0,\quad x(y+z) = xy + xz\quad\text{and}\quad (x+y)z = xy +yz.\]
A \emph{module} over a semiring $S$ is a commutative monoid $(M,+,0)$ together with a scalar multiplication $\o: S\times M\ra M$ such that the following laws hold:
\[
\begin{array}{r@{\ }c@{\ }l@{\qquad}r@{\ }c@{\ }l@{\qquad}r@{\ }c@{\ }l}
  (r+s)a & = & ra+sa, & r(a+b) & = & ra + rb, & (rs)a & = & r(sa), \\
  0a & = & 0, & 1a & = & 1, & r0 & = & 0.
\end{array}
\]

\subsection*{Proof of Proposition \ref{prop:freemon2}}
A constructive proof can be found in \cite{ammu14}. Here we give a more conceptual argument, using the universal property of the tensor product. Observe that the functor $\Psi: \Set\ra\D$ is strongly monoidal, i.e., it preserves the unit and tensor product up to natural isomorphism. Indeed, we have $\Psi 1 = I$ by definition. To see that $\Psi(A\times B) \cong \Psi A \t \Psi B$ for all sets $A$ and $B$, consider the following bijections (natural in $D$):
\begin{align*}
 \D(\Psi(A\times B), D) &\cong \Set(A\times B, \under{D})\\
 &\cong \Set(A,\under{D}^B) \\
 &\cong \D(\Psi A, D^B)\\
 &\cong \D(\Psi A, [\Psi B, D]) \\
 &\cong \D(\Psi A\t \Psi B, D).
\end{align*}
This implies $\Psi(A\times B) \cong \Psi A \t \Psi B$ by the Yoneda lemma. Using the fact that $\Psi$ preserves coproducts, being a left adjoint, we conclude
\[ \xs \cong \coprod_{n<\omega} X^{\cnum{n}} \cong \coprod_{n<\omega} \Psi X_0^n \cong \Psi(\coprod_{n<\omega} X_0^n) = \Psi X_0^*.\]
Alternatively one can show that the right adjoint $\under{\mathord{-}}: \D\ra\Set$ is a monoidal functor. This implies that $\Psi$ preserves free monoids.
\subsection*{Details for Remark~\ref{rem:concept}}

First we recall the $\D$-monoid structure on $[Q,Q]$. Let $\iota'_Q: Q \to I \t Q$ be the left unit isomorphism. Then the unit $j: I \to [Q,Q]$ and multiplication $m: [Q,Q] \t [Q,Q] \to [Q,Q]$ are the unique morphisms making the following diagram commutative, respectively:
\[
\xymatrix{
  [Q,Q] \t Q \ar[r]^-{\ev} & Q \\
  I \t Q 
  \ar[u]^{j \t Q}
  \ar[ru]_{\iota'_Q}
}
\qquad
\xymatrix@C+1pc{
  [Q,Q] \t Q \ar[r]^-{\ev} & Q \\
  [Q,Q] \t [Q,Q] \t Q
  \ar[u]^{m \t Q}
  \ar[r]_-{[Q,Q]\t \ev}
  &
  [Q,Q] \t Q \ar[u]_{\ev}
}
\]
Now all we have to show is that the morphism in~\refeq{eq:eA} is an $F$-algebra homomorphism. It then follows that it is the unique one $e_Q$. Indeed, the diagram below commutes:
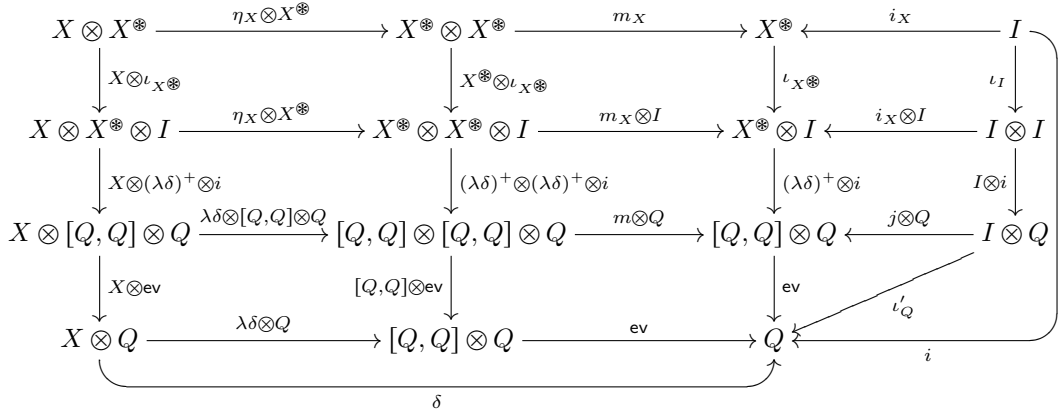
\[
\xymatrix@C+2pc{
  X \t \xs 
  \ar[r]^-{\eta_X \t \xs} 
  \ar[d]^{X \t \iota_{\xs}}
  &
  \xs \t \xs 
  \ar[r]^-{m_X}
  \ar[d]^{\xs \t \iota_{\xs}}
  &
  \xs
  \ar[d]^{\iota_{\xs}}
  &
  I
  \ar[l]_-{i_X}
  \ar[d]_{\iota_I}
  \ar `r[d] `[ddd] [lddd]^(.4){i}
  \\
  X \t \xs \t I
  \ar[r]^-{\eta_X \t \xs}
  \ar[d]^{X \t (\lambda\delta)^+ \t i}
  &
  \xs \t \xs \t I
  \ar[r]^-{m_X \t I}
  \ar[d]^{(\lambda\delta)^+ \t (\lambda\delta)^+ \t i}
  &
  \xs \t I
  \ar[d]^{(\lambda\delta)^+ \t i}
  &
  I \t I
  \ar[l]_-{i_X \t I}
  \ar[d]_{I \t i}
  \\
  X \t [Q,Q] \t Q
  \ar[r]^-{\lambda\delta \t [Q,Q] \t Q}
  \ar[d]^{X \t \ev}
  &
  [Q,Q] \t [Q,Q] \t Q
  \ar[r]^-{m \t Q}
  \ar[d]_{[Q,Q]\t \ev}
  &
  [Q,Q] \t Q
  \ar[d]^{\ev}
  &
  I \t Q
  \ar[l]_-{j \t Q}
  \ar[ld]^{\iota'_Q}
  \\
  X\t Q
  \ar[r]^-{\lambda\delta \t Q}
  &
  [Q,Q]\t Q
  \ar[r]^-{\ev}
  &
  Q
  \ar@{<-} `d[l] `[ll]^\delta [ll]
  &
}
\]

\subsection*{Proof of Proposition \ref{prop:fincoalg}}
Given any coalgebra $(Q,\tau,f)$, consider the morphism $\delta = (X\t Q \xra{\cong} Q\t X \xra{\beta} Q)$ where $\beta$ is the uncurried form of $\tau: Q\ra[X,Q]$, and denote by $\delta^\cst: \xs\t Q\ra Q$ the extension of $\delta$ as in Remark \ref{rem:concept}. We claim that the unique coalgebra homomorphism into $\fc$ is $\lambda h: Q\ra \fc$, where 
\[ h = (Q\t \xs \cong \xs \t Q \xra{\delta^\cst} Q \xra{f} Y).\]
Let us first prove that $h$ is indeed a coalgebra homomorphism. Preservation of outputs is shown by the following commutative diagram:
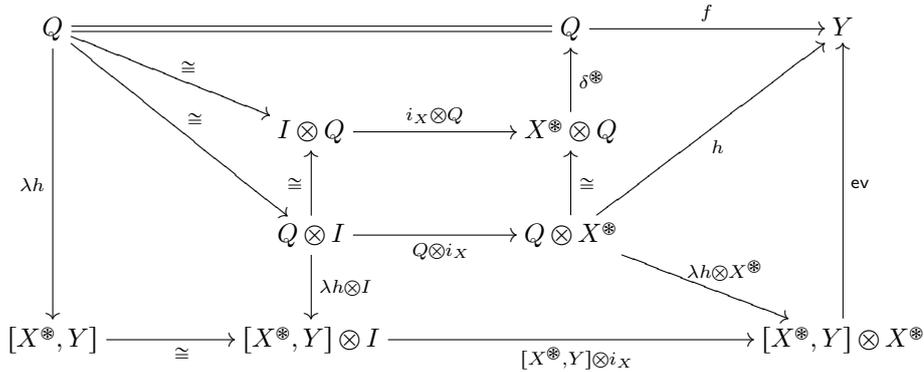
\[
\xymatrix@C+2pc{
Q \ar[ddr]^\cong \ar[rd]^\cong \ar@{=}[rr] \ar[ddd]_{\lambda h} & & Q \ar[r]^{f}  & Y\\
& I \t Q \ar[r]^-{i_X\t Q}& \xs\t Q   \ar[u]_{\delta^\cst}&\\
& Q\t I \ar[r]_{Q\t i_X} \ar[d]^{\lambda h\t I} \ar[u]^\cong & Q\t \xs \ar[dr]^{\lambda h\t \xs} \ar[uur]_h \ar[u]_{\cong}  &\\
\fc \ar[r]_{\cong}& \fc\t I \ar[rr]_{\fc\t i_X} && \fc\t \xs \ar[uuu]_\ev 
}
\]
For preservation of transitions it suffices to show that the following diagram commutes, where $\ol\tau: \fc\t  X\ra \fc$ is the uncurried coalgebra structure of $\fc$:
\[  
\xymatrix{
Q \t X \ar[d]_{\lambda h \t X} \ar[r]^{\beta} & Q \ar[d]^{\lambda h}\\
\fc \t X \ar[r]_{\ol\tau} & \fc
}
\]
But the above diagram is precisely the curried version of the following one (where we omit writing $\otimes$ for space reasons):
\[
\xymatrix{
Q X \ar@{}[ddrrrr]|{(\ast)} \xs \ar[rrrrr]^{\beta\xs} \ar[ddr]_{Q\eta_X\xs} \ar[ddd]_{\lambda h X\xs} & & & && Q \xs \ar[dl]_{\cong} \ar[ddd]^h\\
&&&& \xs Q \ar[d]^{\delta^\cst} &\\
& Q \xs \xs \ar[d]^{\lambda h\xs\xs} \ar[r]^{Q m_X} & Q \xs \ar[drrr]^h \ar[d]^{\lambda h \xs} \ar[r]^\cong & \xs Q \ar[r]^{\delta^\cst} & Q \ar[dr]^{f} &\\
\fc X\xs \ar[r]^-{\begin{turn}{90}$\labelstyle\fc\eta_X\xs$\end{turn}} & \fc \xs\xs \ar[r]^(.6){\begin{turn}{45}$\labelstyle\fc m_X$\end{turn}} & \fc\xs \ar[rrr]_{\ev} & && Y
}
\]
The part $(\ast)$ follows easily from the definition of $\delta^\cst$, and the other parts are clear.
Now suppose that any coalgebra homomorphism of $\lambda h:Q\ra \fc$ is given. We show that $h: Q\t \xs\ra Y$ is determined by the composites $(Q \t X^{\t n} \xra{Q\t i_n} Q\t X^\cst \xra{h} Y)$, $n<\omega$,
where $i_n: X^{\t n}\ra \xs$ is the $n$-th coproduct injection. This proves the uniqueness of $\lambda h$: since $\t$ preserves coproducts, the morphisms $(Q \t i_n)_{n<\omega}$ form a coproduct cocone. For $n=0$, the claim is proved by the diagram
\[
\xymatrix{
Q\t I \ar[ddd]_{Q\t i_0} \ar[rrr]^{\cong} \ar[dr]^{\lambda h \t I} & & & Q \ar[ddd]^{f} \ar[dl]_{\lambda h}\\
& \fc \t I \ar@<2pt>[r]^-\cong \ar[d]_{\fc \t i_0} & \fc \ar@<2pt>[l]^-\cong \ar[rdd]_{f_{[\xs, Y]}}&\\
& \fc\t \xs \ar[drr]_{\ev} & &\\
Q\t \xs \ar[rrr]_h \ar[ur]^{\lambda h \t \xs} &&& Y 
}
\]
 And the following diagram shows that $h\o (Q\t i_{n+1})$ is determined by $h\o (Q\t i_n)$ (again we omit $\t$, in particular we write $X^n$ for $X^\cnum{n}$):
 \[
 \xymatrix{
 QXX^n \ar[dddd]_{\beta X^n} \ar[rrr]^{Qi_{n+1}} \ar[dr]^(.6)*+{\labelstyle \lambda hXX^n} & & & Q\xs \ar[rr]^h \ar[d]^{\lambda h \xs} & & Y\\
   & \fc XX^n \ar[dd]_{\ol\tau X^n}  \ar[rr]^{\fc i_{n+1}} \ar[dr]^(.6)*+{\labelstyle \fc Xi_n} & & \fc\xs \ar[urr]^\ev & & \\
   & & \fc X \xs \ar[r]^(.65){\begin{turn}{25}$\labelstyle\fc \eta_X \xs$\end{turn}} \ar[dr]_{\ol\tau\xs} & \fc\xs\xs \ar[u]_{\fc m_X} \ar@{}[d]|{(*)}& &\\
   & \fc X^n \ar[rr]_{\fc i_n} & & \fc\xs \ar[uuurr]_\ev & &\\
 QX^n \ar[ur]_{\lambda h X^n} \ar[rrrrr]_{Qi_n} & & & & & Q\xs \ar[ull]_{\lambda h \xs} \ar[uuuu]_h  
 }
 \]
Note that part $(*)$ commutes by the definition of the coalgebra structure on $[X^\cst, Y]$ and all other parts are easy to see.

\subsection*{Proof of Proposition \ref{prop:lang}}
Recall from the proof of Proposition \ref{prop:fincoalg} that the uncurried version of $c_Q$ is the morphism
\[ Q\t \xs \cong \xs \t Q \xra{\delta^\cst} Q \xra{f} Y.\]
Hence $c_Q \o i: I\ra \fc$ determines the language
\[ \xs\cong \xs\t I \xra{\xs\t i} \xs\t Q \xra{\delta^\cst} Q \xra{f} Y,\]
and this is precisely $L_Q$, as shown by the diagram below (where $r_X$ is the $F$-algebra structure of $\xs$):
\[
\xymatrix@C+2pc{
\xs   \ar `u[r] `[rrr]^\id [rrr]   \ar[r]^-\cong & \xs\t I \ar[r]^{\xs \t i_X} \ar[dr]_-{\xs\t i} & \xs\t \xs \ar[d]^{\xs \t e_Q} \ar[r]^-{r_X^\cst} & \xs \ar[d]^{e_Q} \ar[r]^{L_Q} & Y\\
 &  & \xs \t Q \ar[r]_{\delta^\cst} & Q \ar[ur]_{f} &
}
\]
This completes the proof.

\subsection*{Proof of Lemma \ref{lem:recog}}
  Let $e: \xs \epito M$ be an $X$-generated $\D$-monoid and let $f: M \to Y$. Then this recognizes the language $L = f \o e$. We are done once we prove that $e$ is the unique $F$-algebra morphism from $\xs$ to the $F$-algebra associated to $M$ and $e \o \eta_X$ (cf.~Remark~\ref{rem:Lalg}). Recall from Proposition~\ref{prop:initalg} that the initial $F$-algebra is the $F$-algebra associated to the free $\D$-monoid $\xs$ and $\eta_X$. Then the following diagram clearly commutes since $e$ is a $\D$-monoid morphism:
  \[
  \xymatrix@C+3pc{
    F\xs = I + X \t \xs 
    \ar[r]^-{I + \eta_X \t \xs}
    \ar[d]_{Fe = I + X \t e}
    &
    I + \xs \t \xs 
    \ar[r]^-{[i_X,m_X]}
    \ar[d]_{I + e \t e}
    &
    \xs
    \ar[d]^e 
    \\
    FM = I + X \t M
    \ar[r]_-{I + (e\o \eta_X) \t M}
    &
    I + M \t M
    \ar[r]_-{[i,m]}
    &
    M
    }
  \]
  This completes the proof.
  
\subsection*{Proof of Lemma \ref{lem:monadic}}
  First of all we know from Lack~\cite[Theorem~2]{lack} that the forgetful functor $U: \DMon \to \D$ is monadic, and by Proposition \ref{prop:freemon} the monad is $\Id^\cst = \coprod_{n<\omega} \Id^{\cnum{n}}$. It suffices to show that $\Id^\cst$ preserves reflexive coequalizers: it then follows that $U$ preserves (in fact, creates) them. For that it is sufficient to prove that each $\Id^{\cnum{n}}$ preserves reflexive coequalizers. This follows from (the proof of)~\cite[Lemma~1]{lack} using that in our setting both $X\t -$ and $-\t X$ preserve all colimits (being left-adjoints). 

\subsection*{Proof of Lemma \ref{lem:canalg}}
Observe that $E = \bigcap E_{x,y}$ where for fixed $x, y \in \xs$, $E_{x,y}$ is the kernel of the $\D$-morphism
   \begin{equation}
     \label{eq:kermor}
     \xymatrix@1{\xs \ar[r]^-{x \bullet -} & \xs \ar[r]^-{-\bullet y} & \xs \ar[r]^-L & Y}.
   \end{equation}
   Since all limits in $\D$ (and in particular, kernels and intersections) are created by the forgetful functor $\under{\mathord{-}}: \D \to \Set$ we see that $E$ has a canonical structure of a $\D$-object as desired. 

\subsection*{Proof of Lemma \ref{lem:mon}}
  It suffices to prove that pairs in $E$ are closed under the monoid operation of $\xs$. 
  Clearly $(\epsilon, \epsilon) \in E$, where $\epsilon: 1 \to \xs$ is the adjoint transpose of the unit $I =\Psi 1 \to \xs$ of the free $\D$-monoid on $X$. 
  Given $(u,v)$ and $(u',v')$ in $E$ we show that $(u\bullet u', v \bullet v')$ is in $E$, too. Indeed, we have for all $x, y \in \xs$ that
  \[
  L(x\bullet u\bullet u'\bullet y) = L(x\bullet v\bullet u'\bullet y) = L(x\bullet v\bullet v'\bullet y).\qedhere
  \]
  
\subsection*{Proof of Lemma \ref{lem:ker}}
To see that $E$ from Definition~\ref{d:syn} satisfies this property let $l', r': E' \to \xs$ be any $\D$-monoid congruence contained in $K$ via $m': E' \monoto K$ with $l_0 \o m' = l'$ and $r_0 \o m' = r'$. Since $l'$ is a $\D$-monoid morphism it is easy to see that for every $x: I \to \xs$ the following square commutes (note that $(x,x)\in E'$ since $E'$ is reflexive):
  \[
  \xymatrix{
    E'
    \ar[d]_{l'}
    \ar[r]^-{(x,x) \bullet -}
    &
    E'
    \ar[d]^{l'}
    \\
    \xs \ar[r]_-{x \bullet -}
    &
    \xs
  }
  \]
  and similarly for $r'$ in lieu of $l'$ and/or $- \bullet y$ in lieu of $x \bullet -$. It follows that the morphism \refeq{eq:kermor} in the proof of Lemma \ref{lem:canalg} merges $l'$ and $r'$:
  \begin{align*}
    L \o (- \bullet y) \o (x \bullet -) \o l' 
    & = L \o l' \o (- \bullet (y,y) \o ((x,x) \bullet -) \\
    & = L \o l_0 \o m' \o (- \bullet (y,y)) \o ((x,x) \bullet -) \\
    & = L \o r_0 \o m' \o (- \bullet (y,y)) \o ((x,x) \bullet -) \\
    & = L \o r' \o (- \bullet (y,y)) \o ((x,x) \bullet -) \\
    & = L \o(- \bullet y) \o (x \bullet -) \o r' 
  \end{align*}
 Hence $E'$ is contained in any $E_{x,y}$, and therefore it is contained in their intersection $E$. 

\subsection*{Proof of Theorem \ref{thm:syn}}
  \takeout{ 
  Note that once we have a well-defined map $h: M \to \Syn L$ we also know that $h$ is a $\D$-monoid morphism by the homomorphism theorem. 

  Define $h$ by $h(e(u)) = e_L(u)$ for every $u \in \Psi\Sigma^*$. Then this is well-defined, for suppose that $e(u) = e(v)$ for $u, v \in \Psi\Sigma^*$. Then $L(u) = L(v)$ since $L = f \cdot e$. Now let $x,y \in \Psi\Sigma^*$. Then we have
\[
e(x \bullet u \bullet v) = e(x)e(u)e(y) = e(x)e(v)e(y) = e(x \bullet u \bullet y)
\]
and therefore $L(x \bullet u \bullet v) = L(x \bullet u \bullet y)$. Hence $(u,v) \in E$ so that $e_L(u) = e_L(v)$ as desired. 

It remains to prove that $f_L \cdot h = f$. To see this we compute
\[
f_L \o h \o e = f_L \o e_L = L = f \o e,
\]
and we use that $e$ is an epimorphism.}

  Suppose that we have an $X$-generated $\D$-monoid $e: \xs \epito M$ and a morphism $f: M \to Y$ recognizing $L$, i.e.~such that $L = f \o e$. Let $l_0,r_0: K \to \xs$ be the kernel pair of $L$ as in Lemma~\ref{lem:ker} and take the kernel pair $l_M, r_M: K_M \to \xs$ of $e$. Now clearly we have 
\[
L \o l_M = f \o e \o l_M = f\o e \o r_M = L \o r_M.
\]
Hence, there is a unique $n: K_M \to K$ such that $l_0 \o n = l_M$ and $r_0 \o n = r_M$. It follows that $n$ is monomorphic and so $K_M$ is a $\D$-monoid congruence contained in $K$. Consequently $K_M$ is contained in $E$ (the largest $\D$-monoid congruence contained $K$ by Lemma \ref{lem:ker}) via $o: K_M \monoto E$ with $l_M = l \o o$ and $r_M = r \o o$. Then we obtain that 
\[
e_L \o l_M = e_L \o l \o o = e_L \o r \o o = e_L \o r_M.
\]
Thus, using that $e: \xs \to M$ is the coequalizer of its kernel pair $l_M, r_M$ we obtain a unique $h: M \to \Syn L$ with $e_L = h \o e$, as desired.

\subsection*{Proof of Proposition \ref{prop:transmon}}
Let $(Q,\delta, i, f)$ be a $\D$-automaton. By definition it accepts the language $L_Q = (\xs \xra{e_Q} Q \xra{f} Y)$ where $e_Q$ is the unique $F$-algebra morphism. Consider the morphism that evaluates any endomorphism of $Q$ at the initial state:
\[
\ev_i = ([Q,Q] \cong [Q,Q]\t I \xra{[Q,Q]\t i} [Q,Q]\t Q \xra{\ev} Q).
\]
Now let 
\[
f_{\T Q} = (\T Q \xra{m_{\T Q}} [Q,Q] \xra{\ev_i} Q \xra{f} Y).
\]
With this morphism $\T Q$ recognizes $L$; indeed, using the canonical isomorphism $\iota_Z: Z \to Z \t I$ we compute:
\begin{align*}
  L_Q & = f \o e_Q \\
    & = f \o \ev \o ((\lambda\delta)^+ \t i) \o \iota_{\xs} & \text{(see Remark~\ref{rem:concept})}\\
    & = f \o \ev_i \o \iota_{[Q,Q]}^{-1}\o ((\lambda\delta)^+ \t I) \o \iota_{\xs} & \text{(def.~of $\ev_i$)} \\
    & = f \o \ev_i  \o (\lambda\delta)^+ & \text{(naturality of $\iota$)} \\
    & = f \o \ev_i \o m_{\T Q} \o e_{\T Q} & \text{(see Definition~\ref{def:TA})} \\
    & = f_{\T Q} \o e_{\T Q} & \text{(def.~of $f_{\T Q}$)}
\end{align*}
This completes the proof.

\subsection*{Proof of Theorem \ref{thm:tran}}
  Let $\Min{L}=(Q,\delta,i,f)$, and write $\delta_x: Q\ra Q$ for $e_{\T{Q}}(x)$ ($x\in\xs$). Note that $\delta_{x\bullet y}= \delta_y\o \delta_x$ for all $x,y\in \xs$ since $e_{\T{Q}}$ is a $\D$-monoid morphism. Observe also that the unique $F$-algebra homomorphism $e_Q: \xs\ra Q$ assigns to $x\in\xs$ the element $\delta_x\o i: I\ra Q$, and the unique $T$-coalgebra homomorphism $m_Q: Q\ra [\xs, Y]$ assigns to a state $q: I\ra Q$ the language $x\mapsto f\o\delta_x\o q$. It suffices to show that the kernel of $e_{\T{Q}}$ is the syntactic congruence of $L$, that is, for all $u,v\in \xs$ one has
\[
\delta_u = \delta_v \quad\text{iff}\quad \forall x,y\in \xs: L(x\bullet u \bullet y)=L(x\bullet v \bullet y).
\]
To see this, we reason as follows:
\begin{align*}
\delta_u = \delta_v &\Lra \forall x: \delta_u \o e_Q(x)= \delta_v\o e_Q(x) & \text{($e_Q$ surjective)}\\
&\Lra\forall x: \delta_u\o \delta_x\o i = \delta_v\o \delta_x \o i & \text{(def. $e_Q$)}\\
& \Lra\forall x: m_Q\o \delta_u\o \delta_x\o i = m_Q\o \delta_v\o \delta_x \o i & \text{($m_Q$ injective)}\\
& \Lra \forall x, y: f \o \delta_y\o \delta_u\o \delta_x\o i = f\o \delta_y\o \delta_v\o \delta_x\o i & \text{(def. $m_Q$)}\\
& \Lra \forall x, y: f\o \delta_{x\bullet u\bullet y}\o i = f \o \delta_{x\bullet v \bullet y} \o i & \text{(def. $\delta_{(\mathord{-})}$)}\\
& \Lra \forall x, y: f\o e_Q(x\bullet u\bullet y) = f \o e_Q(x\bullet v\bullet z) & \text{(def. $e_Q$)}\\
& \Lra \forall x, y: L(x\bullet u \bullet y) = L(x\bullet v \bullet y) & \text{($L=L_Q$)}
\end{align*}

\subsection*{Proof of Theorem \ref{thm:reg}}

\textbf{Remark.} 
\begin{enumerate} \item The functor $FQ=I+X\t Q$ preserves strong epimorphisms because the left adjoint $X\t \mathord{-}$ does and strong epimorphisms are closed under coproducts. Therefore every $F$-algebra homomorphism factorizes into a surjective homomorphism (carried by a strong epimorphism in $\DCat$) and an injective one (carried by a monomorphism in $\DCat$). By the \emph{reachable part} $Q_r$ of an automaton $(Q,\delta,f,i)$ we mean the image of the initial $F$-algebra homomorphism, i.e., $\xymatrix{
e_Q = (\xs \ar@{->>}[r]^>>>>>{e_r} & Q_r \ar@{ >->}[r]^{m_r} & Q). 
}$
Putting $f_r := f\o m_r: Q_r\ra Y$, the $F$-algebra $Q_r$ becomes an automaton, and $m_r$ an automata homomorphism.\qed
\item In the following automata, (co-)algebras and monoids with finitely presentable carrier are referred to as \emph{fp-automata}, \emph{fp-(co-)algebras} and \emph{fp-monoids}, respectively.
\end{enumerate}

Now for the proof of the theorem. (a)$\Lra$(b) follows from Theorem \ref{thm:minaut} and the closure of $\D_{f}$ under subobjects and strong quotients. Similarly, $(c)\Lra(d)$ follows from the universal property of the syntactic monoid (see Definition \ref{def:syn}) and again closure of $\D_{f}$ under subobjects and strong quotients. (c)$\Ra$(a) is a consequence of Lemma \ref{lem:recog}. To prove (a)$\Ra$(c), let $Q$ be any fp-automaton accepting $L$. Then by Proposition \ref{prop:transmon} the transition monoid $\T{Q}\monoto[Q,Q]$ recognizes $L$, so by closure of $\D_{f}$ under subobjects it suffices to show that $[Q,Q]$ is a finitely presentable object of $\D$. Assuming that $Q$ has $n$ generators as an algebra of $\D$, the map $[Q,Q]\ra Q^n$ defined by restriction to the set of generators is an injective $\D$-morphism. Since $\D_{f}$ is closed under subobjects and finite products, it follows that $[Q,Q]$ is finitely presentable.

\subsection*{Proof of Proposition \ref{prop:ratfix}}
We describe mutually inverse maps 
\[
(I\xra{x}\rho T)\mapsto (\xs\xra{L_x} Y) 
\qquad \text{and}\qquad 
(\xs\xra{L} Y) \mapsto (I\xra{x_L} \rho T)
\] 
between the elements of $\rho T$ and the $\D$-regular languages.
Let $h_Q: Q \to \rho T$ be the injections of the colimit $\rho T$, where $Q=(Q,\delta_Q, f_Q)$ ranges over all fp-coalgebras. Note that this colimit is filtered since $\DCat_f$ is closed under finite colimits. Moreover, since colimits of coalgebras are formed in the underlying category, the morphisms $h_Q$ also form a colimit cocone in $\DCat$. Given an element $I\xra{x} \rho T$ of the rational coalgebra we define a $\D$-regular language $L_x: \xs \ra Y$ as follows: since $I=\Psi 1$ is finitely presentable, there exists an fp-coalgebra $Q$ and a morphism $i_Q: I\ra Q$ such that $x=h_Q\o i_Q$. For the $F$-algebra $(Q,\delta_Q, i_Q)$ we have the unique $F$-algebra homomorphism $e_{Q}: \xs \ra Q$, and we put $L_x := f_Q\o e_{Q}$. 

We need to show that $L_x$ is well-defined, that is, for any other factorization $x = h_{Q'}\o i_{Q'}$  we have $f_Q\o e_{Q}=f_{Q'}\o e_{Q'}$. Given such a factorization, since the $h_Q$ form a filtered colimit, there exists an fp-coalgebra $Q''=(Q'',\delta_{Q''},f_{Q''})$ and  coalgebra homomorphisms $h_{QQ'}: Q\ra Q'$ and $h_{Q'Q''}: Q'\ra Q''$ with $h_{QQ'}\o i_Q = h_{Q'Q''}\o i_{Q'} =: i_{Q''}$. Then for the $F$-algebra $(Q'', \delta_{Q''}, i_{Q''})$ we have the unique homomorphism $e_{Q''}: \xs \ra Q''$. Moreover, $h_{QQ'}$ and $h_{Q'Q''}$ are also homomorphisms of $F$-algebras. If follows that
\begin{align*}
f_Q \o e_{Q} &= f_{Q''} \o h_{QQ''} \o e_{Q} & \text{($h_{QQ''}$ coalgebra homomorphism)}\\
&= f_{Q''} \o e_{Q''} & \text{($h_{QQ''}$ $F$-algebra hom., $\xs$ initial)}
\end{align*} 
and analogously $f_{Q'} \o e_{Q'}=f_{Q''} \o e_{Q''}$. Hence $f_Q\o e_{Q}=f_{Q'}\o e_{Q'}$, as claimed.

Conversely, let $L:\xs \ra Y$ be a $\D$-regular language. Then there exists an fp-automaton $(Q,\delta_Q,i_Q,f_Q)$ with $L = f_Q\o e_Q$. Put $x_L := h_Q\o i_Q: I \ra \rho T$. To prove the well-definedness of $x_L$, consider the automata homomorphisms 
\[ \xymatrix{
Q & Q_r \ar@{ >->}[l]_-m \ar@{-->>}[r]^-e & \Min{L}
}  \]
of Theorem \ref{thm:minaut}.  Then
\begin{align*}
h_Q\o i_Q &= h_Q \o m\o i_{Q_r} & \text{($m$ algebra hom.)}\\
&= h_{Q_r} \o i_{Q_r} & \text{($(h_Q)$ cocone and $m$ coalgebra hom.)}\\
&= h_{\Min{L}} \o e \o i_{Q_r} & \text{($(h_Q)$ cocone and $e$ coalgebra hom.)}\\
&= h_{\Min{L}} \o i_{\Min{L}} & \text{($e$ algebra hom.)}
\end{align*}
Hence $x_L = h_Q\o i_Q$ is independent of the choice of $Q$. It now follows immediately from the definitions that $x\mapsto L_x$ and $L\mapsto x_L$ are mutually inverse and hence define the desired bijective correspondence.

\section{Dual Characterization of Syntactic Monoids}\label{app:dual}
Here we give a more detailed account of the dual view of syntactic monoids indicated in Section \ref{sec:rat}. This section are largely based on results from our papers \cite{ammu14,ammu15} where a categorical generalization of Eilenberg's variety theorem was proved. We work with the following

\begin{assumptions}
From now on $\D$ is a locally finite entropic variety whose epimorphisms are surjective.
Moreover, we assume that there is another locally finite variety $\C$ such that the full subcategories $\C_{f}$ and $\D_f$ of finite algebras are dually equivalent. (Two such varieties $\Cat$ and $\DCat$ are called \emph{predual}.)
\end{assumptions}

The action of the equivalence functor $\C_{f}^{op} \xra{\simeq} \D_f$ on objects and morphisms is written $Q\mapsto \widehat{Q}$ and $f\mapsto \widehat{f}$. Letting $\ol{I}\in \C_f$ denote the free one-generated object of $\C$ we choose the output object $Y\in\D_f$ to be the dual object of $\ol{I}$. Moreover, let $\ol Y\in \C_f$ be the dual object of $I\in \D_f$, the free one-generated object of $\DCat$.  Finally, we put $X = \Psi X_0$ for a finite alphabet $X_0$. Note that the underlying sets of $\ol{Y}$ and $Y$ are isomorphic:
\[\under{\ol{Y}}\cong \C(\ol{I},\ol{Y}) \cong \D(I,Y) \cong \under{Y}.\]
To simplify the presentation, we will assume in the following that $\ol{Y}$ and $Y$ have a two-element underlying set $\{0,1\}$. This is, however, inessential -- see Remark \ref{rem:twoel} at the end of this section.

\begin{example}
The categories $\C$ and $\D$ in the table below satisfy our assumptions.
\begin{center}
\begin{tabular}{|ll|}
\hline\rule[11pt]{0pt}{0pt}
$\C$ & $\D $ \\
\hline
$\BA$ & $\Set$ \\
$\BR$ & $\PSet$\\
$\JSL$ & $\JSL$ \\
$\Vect{\Int_2}$ & $\Vect{\Int_2}$ \\
\hline
\end{tabular}
\end{center}
The case $\BA$/$\Set$ is a restriction of Stone duality: the dual equivalence functor $\BA_f^{op}\xra{\simeq} \Set_f$ assigns to a finite boolean algebra $B$ the set $\At(B)$ of its atoms, and to a homomorphism $h: A\ra B$  the map $\At(h): \At(B)\ra\At(A)$ sending $b\in \At(B)$ to the unique atom $a\in\At(A)$ with $ha\geq b$. Using a similar Stone-type duality, we proved in \cite{ammu15} that the the category
$\BR$ of non-unital boolean rings (i.e., rings without $1$ satisfying the equation $x\o x= x$) is predual to $\PSet$. The other two examples above correspond to the well-known self-duality of finite semilattices and finite-dimensional vector spaces, respectively. We refer to \cite{ammu15} for details.
\end{example}

On $\C$ we consider the endofunctor $\ol{T}Q = \ol{Y}\times Q^{X_0}$. Its coalgebras are precisely deterministic automata in $\C$ without an initial state, represented as triples $(Q,\gamma_a,f)$ with transition morphisms $\gamma_a: Q\ra Q$ ($a\in X_0$) and an output morphism $f: Q\ra \ol{Y}$.

\begin{example}
	 In $\Cat=\BA$ a $\ol{T}$-coalgebra is a deterministic automaton with a boolean algebra $Q$ of states, boolean transitions morphisms $\gamma_a$, and an output map $f: Q\ra \{0,1\}$ which specifies (via the preimage of $1$) an ultrafilter $F\seq Q$ of final states.
\end{example}

The rational coalgebra $\rho \ol{T}$ of $\ol{T}$ (i.e., the colimit of all finite $\ol T$-coalgebras) has as states the regular languages over $X_0$. The final state predicate $f: \rho\ol T\ra \ol Y=\{0,1\}$ sends a language to $1$ iff it contains the empty word $\epsilon$, and the transitions $\gamma_a: \rho \ol T \ra \rho\ol T$ for $a\in X_0$ are given by $\gamma_a(L)=a^{-1} L$.
Here $a^{-1}L = \{w\in X_0^*: aw\in L\}$ denotes the \emph{left derivative} of $L$ w.r.t. the letter $a$. Similarly, the \emph{right derivatives} of $L$ are defined by $La^{-1} = \{w\in X_0^*: wa\in L\}$ for $a\in X_0$. 

\begin{example}
In $\Cat=\BA$ the rational $\ol{T}$-coalgebra is the boolean algebra of all regular languages over $X_0$ (w.r.t. union, intersection and complement), equipped with the above transitions and final states. Note that the transition map $a^{-1}(\mathord{-})$ is indeed a boolean homomorphism because left derivatives preserve all boolean operations. Moreover, the final states -- viz. the set of all regular languages containing the empty word -- form a (principal) ultrafilter.	
\end{example}

The coalgebra $\rho\ol T$ is characterized by a universal property: every finite $\ol T$-coalgebra has a unique coalgebra homomorphism into it (which sends a state to the language it accepts in the classical sense of automata theory). A finite $\ol T$-coalgebra is called a \emph{subcoalgebra} of $\rho \ol T$ if this unique morphism is injective, i.e., distinct states accept distinct languages. In \cite{ammu14} we related finite $\ol T$-coalgebras in $\C$ to finite $F$-algebras in the predual category $\D$. Note that $X=\Psi X_0$ implies $FA=I+ X\t A\cong I+\coprod_{X_0} A$, so $F$-algebras $F$-algebras $(A,\delta)$  can be represented as triples $(A,\delta_a,i)$ with $\delta_a: A\ra A$ ($a\in X_0$) and $i: I\ra A$. They correspond to automata in $\DCat$ with inputs from the alphabet $X_0$ and without final states.

\begin{proposition}[see \cite{ammu14}]
\begin{enumerate}[(a)]\item  The categories of finite $\ol{T}$-coalgebras and finite $F$-algebras are dually equivalent. The equivalence maps any finite $\ol{T}$-coalgebra $Q=(Q,\gamma_a,f)$ to its \emph{dual $F$-algebra} $\widehat Q = (\widehat Q, \widehat{\gamma_a},\widehat f)$:
\[ (\ol{Y} \xleftarrow{f} Q\xra{\gamma_a} Q) \quad\mapsto\quad (I\xra{\widehat{f}}\widehat{Q}\xleftarrow{\widehat{\gamma_a}} \widehat{Q}). \]
\item A finite $\ol{T}$-coalgebra $Q$ is a subcoalgebra of $\rho \ol{T}$ iff its dual $F$-algebra $\widehat{Q}$ is a quotient of the initial $F$-algebra $\xs$.
\end{enumerate}
\end{proposition}

\begin{example}
For a finite $\ol T$-coalgebra $(Q,\gamma_a,f)$ in $\BA$ the dual $F$-algebra $\widehat Q$ has as states the atoms of $Q$, and the initial state is the unique atomic final state of $Q$. Moreover, there is a transition $z\xra{a} z'$ for $a\in X_0$ in $\widehat Q$ iff $z'$ is the unique atom with $\gamma_a(z')\geq z$ in $Q$.	
\end{example}

 By a \emph{local variety of languages over $X_0$ in $\Cat$} we mean a subcoalgebra $V$ of $\rho \ol T$ closed under right derivatives (i.e. $L\in\under{V}$ implies $La^{-1}\in\under{V}$ for all $a\in X_0$). Note that a local variety is also closed under the $\Cat$-algebraic operations of $\rho \ol T$, being a sub\emph{algebra} of $\rho\ol T$ in $\Cat$, and under left derivatives, being a sub\emph{coalgebra} of $\rho \ol T$. 
 
 \begin{example}
 	A local variety of languages in $\BA$ is a set of regular languages over $X_0$ closed under the boolean operations (union, intersection and complement) as well as left and right derivatives. This concept was introduced by Gehrke, Grigorieff and Pin \cite{ggp08}.
\end{example}
In the following proposition recall that every $X_0$-generated $\DCat$-monoid defines an $F$-algebra, see Definition \ref{def:asso}.

\begin{proposition}[see \cite{ammu14}]\label{prop:locvarmon}
  A finite subcoalgebra $V$ of $\rho \ol{T}$ is a local variety iff its dual $F$-algebra $\widehat{V}$ is derived from an $X_0$-generated $\D$-monoid.
\end{proposition}

In other words, given a finite local variety $V\monoto \rho \ol T$ in $\Cat$,  there exists a unique $\D$-monoid structure on $\widehat V$ making the unique (surjective) $F$-algebra homomorphism $e_{\widehat{V}}: \xs\ra \widehat{V}$ is $\D$-monoid morphism. Hence the monoid multiplication on $\widehat V$ is (well-)defined by $e_{\widehat{V}}(x)\bullet e_{\widehat{V}}(y) := e_{\widehat{V}}(x\bullet y)$ for all $x,y\in\xs$, and the unit it the initial state of the $F$-algebra $\widehat V$.

\begin{rem}
A \emph{pointed} $\ol T$-coalgebra is a $\ol T$-coalgebra $(Q,\gamma_a,f)$ equipped with an initial state $i: \ol I\ra Q$. Observe that every finite pointed $\ol T$-coalgebra $(Q,\gamma_a,f,i)$ dualizes to a finite $\DCat$-automaton $(\widehat Q,\widehat \gamma_a, \widehat f, \widehat i)$. The \emph{language} of $(Q,\delta_a, f,i)$ is the function
\[L_Q: X_0^*\ra \under{\ol{Y}},\quad a_1\dots a_n \mapsto f\o\delta_{a_n}\o\dots\o \delta_{a_1}\o i.\]
Letting $m_Q: Q\ra \rho \ol T$ denote the unique coalgebra homomorphism, $L_Q$ is precisely the element of $\rho \ol T$ determined by $I\xra{i} Q \xra{m_Q} \rho \ol T$.
Since $\under{\ol{Y}}=\under{Y}$ and $\xs=\Psi X_0^*$,  the function $L_Q: X_0^*\ra \under{\ol Y}$ can be identified with its adjoint transpose $L_Q^@: \xs \ra Y$, i.e., with a language in $\DCat$. The \emph{reversal} of a language $L: \xs\ra Y$ in $\D$ is $L^\rev = L\o \rev: \xs\ra Y$, where $\rev: \xs \ra \xs$ denotes the unique morphism of $\D$ extending the function $X_0^*\ra X_0^*$ that reverses words.
\end{rem}

\begin{proposition}[see \cite{ammu15}]
The language accepted by a finite pointed $\ol{T}$-coalgebra is the reversal of the language accepted by its dual $\D$-automaton .
\end{proposition}

If a finite $X_0$-generated $\DCat$-monoid $e: \xs \ra M$ recognizes a language $L: \xs\ra Y$ via $f: M\ra Y$, i.e., $L=(\xymatrix{\xs \ar@{->>}[r]^e& M \ar[r]^f & Y})$, we dually get the morphism $\xymatrix{\ol I \ar[r]^i & V \ar@{>->}[r]^m & \rho \ol T}$ (where $V$ is the local variety dual to $M$, $i$ is the dual morphism of $f$ and $m$ is the unique coalgebra homomorphism) choosing the element $L^\rev$ of $\rho \ol T$. Now suppose that $L$ is a regular language, and let $V_L$ be the finite local variety of languages dual to the syntactic $\D$ monoid $\Syn{L}$, see Proposition \ref{prop:locvarmon}. The universal property of $\Syn{L}$ in Definition \ref{def:syn} then dualizes as follows: $V_L$ is (i) a local variety containing $L^\rev$, and (ii) for every local variety  $V\monoto\rho \ol{T}$ containing $L^\rev$, the local variety $V_L$ is contained in $V$. In other words, $V_L$ is the \emph{smallest} local variety containing $L^\rev$.
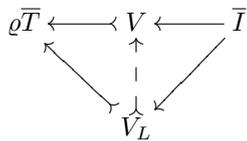
\[
\xymatrix{
  \rho \ol{T}
  &
 V \ar@{ >->}[l]
  &
  \ol{I} \ar[l] \ar[dl]
  \\
  &
  V_L \ar@{ >->}[ul] \ar@{>-->}[u]
}
\]
In summary, we have proved following dual characterization of syntactic $\D$-monoids:
\begin{theorem}\label{thm:syndualapp}
For every regular language $L$ the syntactic $\D$-monoid $\Syn L$ is dual to the smallest local variety of languages over $X_0$ in $\C$ containing $L^\rev$.
\end{theorem}

\begin{example} For $\Cat=\BA$ and $\D=\Set$ the previous theorem gives the following construction of the syntactic monoid of a regular language $L\seq X^*$:
\begin{enumerate}
	\item Form the smallest local variety of languages $V_L\seq \rho\ol T$ containing $L^\rev$. Hence $V_L$ is the closure of the (finite) set of all both-sided derivatives $u^{-1}L^\rev v^{-1} = \{w\in X^*: uwv\in L^\rev\}$ ($u,v\in X^*$)  under union, intersection and complement.
	\item Compute the $F$-algebra $\widehat{V_L}$ dual to the coalgebra $V_L$. The states of $\widehat{V_L}$ are the atoms of $V_L$, and the initial state is the unique atom $i\in V_L$ containing the empty word. Given atoms $z,z'\in V_L$ and $a\in X$, there is a transition $z\xra{a}z'$ in $\widehat{V_L}$ iff $z'$ is the (unique) atom with $z\subseteq a^{-1}z'$.
	\item Define a monoid multiplication on $\widehat{V_L}$ as follows: given states $z,z'\in\widehat{V_L}$, choose words $w,w'\in X^*$ with $i\xra{w} z$ and $i\xra{w'} z'$ in $\widehat{V_L}$. Then $z\bullet z'$ is the state reached on input $ww'$, i.e., $i\xra{ww'} z\bullet z'$. The resulting monoid (with multiplication $\bullet$ and unit $i$) is $\Syn{L}$.
\end{enumerate}
\end{example}

By dropping right derivatives and using the correspondence between finite subcoalgebras of $\rho \ol T$ and finite quotient algebras of $\xs$, one also gets the following dual characterization of minimal $\DCat$-automata:

\begin{theorem}\label{thm:mindual}
For every regular language $L$ the minimal $\D$-automaton for $L$ is dual to the smallest subcoalgebra of $\rho \ol T$ containing $L^\rev$
\end{theorem}

\begin{rem}\label{rem:twoel}
Our above assumption that $Y$ and $\ol Y$ have two elements is inessential. Without this assumption, the rational coalgebra $\rho\ol T$ is not carried by regular languages, but more generally by \emph{regular behaviors}, i.e, functions $b:X_0^*\ra \under{Y}$ realized by finite Moore automata with output set $\under{Y}$. The coalgebra structure is given by the output map $b\mapsto b(\epsilon)$, and transitions $b\xra{a} a^{-1}b$ for $a\in X_0$, where $a^{-1}b$ is the (generalized) left derivative of $b$ defined by $a^{-1}b(w) = b(aw)$. (Right derivatives are defined analogously.) A \emph{local variety of behaviors over $X_0$ in $\Cat$} is a subcoalgebra of $\rho \ol T$ closed under right derivatives. All results of this section hold for this more general setting, see Section 5 of \cite{ammu15} for details. In particular, this allows us to cover the case $\C=\D=\Vect{\K}$ for an arbitrary finite field $\K$. In this case Theorem \ref{thm:syndual} states that the syntactic associative algebra of a rational weighted language $L: X_0^*\ra \K$ dualizes to the smallest set of rational weighted languages that contains $L^\rev$ and is closed under scalar multiplication, addition and left and right derivatives.
\end{rem}
\fi 

\end{document}

%% file: preamble.tex
\usepackage{amsmath, amsthm, amssymb, bbm}
\usepackage{dsfont}
\usepackage{xspace}
\usepackage{mathabx}
\usepackage{tikz}
\usetikzlibrary{shapes,snakes}

\newcommand*\cnum[1]{{\otimes #1}}
\newcommand*\ecnum[1]{{\otimes(#1)}}

\csname@@input\endcsname xypic \xyoption{arrow} \xyoption{tips}
\SelectTips{cm}{10} \UseTips
\newdir{ >}{{}*!/-5pt/\dir{>}}

\DeclareFontFamily{U}{cmr}{} \DeclareFontShape{U}{cmr}{a}{n}{<5>
cmb5<7>cmb7<10>cmb10}{}

\DeclareFontFamily{U}{msa}{} \DeclareFontShape{U}{msa}{m}{n}{<-6
>msam5<6-8.5>msam7<8.5->msam10}{} \DeclareSymbolFont{latex-font
msa}{U}{msa}{m}{n} \DeclareFontFamily{U}{msb}{}
\DeclareFontShape{U}{msb}{m}{n}{<-6>msbm5<6-8.5>msbm7<8.5->msbm1
0}{} \DeclareSymbolFont{latex-font msb}{U}{msb}{m}{n}

\DeclareMathAlphabet{\mathbt}{OT1}{cmr}{a}{n}

\theoremstyle{plain}
\newtheorem{proposition}[theorem]{Proposition}

\theoremstyle{definition}
\newtheorem{assumptions}[theorem]{Assumptions}
\newtheorem{notation}[theorem]{Notation}
\newtheorem{construction}[theorem]{Construction}
\newtheorem{rem}[theorem]{Remark}

\DeclareMathSymbol\emptyset\mathord{latex-font msb}{"3F}
\DeclareMathSymbol\onto\mathrel{latex-font msa}{"10}
\DeclareMathSymbol\square\mathrel{latex-font msa}{"03}
\DeclareMathAccent{\bigwidehat}{\mathord}{latex-font msb}{"5B}

\def\D{\DCat}
\def\C{\Cat}

\def\refeq#1{(\ref{#1})}
\def\epsilon{\varepsilon}

\def\At{\mathsf{At}}

\renewcommand{\rho}{\varrho}
\def\ol{\overline}

\def\wt{\widetilde}

\def\o{\cdot}

\def\DMon{\mathbf{Mon}(\DCat)}
\newcommand{\cst}{{\scriptscriptstyle\oasterisk}}
\def\S{\mathds{S}}

\renewcommand{\t}{\otimes}
\newcommand{\takeout}[1]{\empty}

\renewcommand{\phi}{\varphi}
\newcommand{\ra}{\rightarrow}
\newcommand{\xra}{\xrightarrow}

\newcommand{\Ra}{\Rightarrow}

\newcommand{\Lra}{\Leftrightarrow}
\newcommand{\seq}{\subseteq}

\newcommand{\tl}{\widetilde}

\newcommand{\T}[1]{\mathsf{T}(#1)}

\newcommand{\Min}[1]{\mathsf{Min}(#1)}

\newcommand{\Ab}{\mathbf{Ab}}
\newcommand{\Inv}{\mathbf{Inv}}

\newcommand{\Set}{\mathbf{Set}}

\newcommand{\BA}{\mathbf{BA}}
\newcommand{\BR}{\mathbf{BR}}

\newcommand{\Cat}{\mathscr{C}}
\newcommand{\ev}{\mathsf{ev}}
\newcommand{\DCat}{\mathscr{D}}

\newcommand{\id}{\mathsf{id}}
\newcommand{\Id}{\mathsf{Id}}

\newcommand{\Mon}[1]{\mathbf{Mon}(#1)}

\newcommand{\Vect}[1]{\mathbf{Vec}(#1)}
\newcommand{\SMod}[1]{\mathbf{Mod}(#1)}

\newcommand{\JSL}{{\mathbf{JSL}_0}}

\newcommand{\PSet}{\mathbf{Set}_\bot}

\newcommand{\under}[1]{|#1|}

\newcommand{\epito}{\twoheadrightarrow}

\newcommand{\monoto}{\rightarrowtail}

\newcommand{\rev}{\mathsf{rev}}

\newcommand{\Pow}{\mathcal{P}}

\newcommand{\Syn}[1]{\mathsf{Syn}(#1)}

\newcommand{\K}{\mathds{K}}

\newcommand{\Int}{\mathds{Z}}

\newcommand{\Reg}{\mathsf{Reg}}

\newcommand{\fc}{[X^\cst,Y]}
\newcommand{\xs}{X^\cst}
